\documentclass[10pt,a4paper,english]{article}

\usepackage{amsmath}
\usepackage{amscd}
\usepackage{latexsym}
\usepackage[english]{babel}
\usepackage{amsfonts}
\usepackage{version}
\usepackage{graphicx}
\usepackage{changepage}
\usepackage[T1]{fontenc}
\usepackage[utf8]{inputenc}
\usepackage{amsthm}
\usepackage{amssymb}
\usepackage{version}
\usepackage{fancyhdr}
\usepackage{caption}
\usepackage{stmaryrd}
\usepackage{url}
\usepackage{hyperref}
\usepackage{graphicx}
\usepackage{subfigure}

\newcommand{\esp}[2][\mathbb E] {#1\left[#2\right]}

\numberwithin{equation}{section}

\newcommand{\condespf}[2][\F_t]       {\mathbb E\left.\left[#2\right|#1\right]}

\newcommand{\condespg}[2][\G_t]       {\mathbb E\left.\left[#2\right|#1\right]}

\newcommand{\ud}{\mathrm d}

\newcommand{\norm}[1]		{\left\| #1 \right \|}

\newcommand{\R}{\R}

\newtheorem{theorem}{Theorem}[section]
\newtheorem{lemma}[theorem]{Lemma}
\newtheorem{defn}[theorem]{Definition}

\newtheorem{prop}[theorem]{Proposition}

\newtheorem{rem}[theorem]{Remark}

\def \H {{\mathcal H}}
\def \Hb {{\mathbb H}}
\def \G {{\mathcal G}}
\def \Gb {{\mathbb G}}
\def \F {{\mathcal F}}
\def \Fb {{\mathbb F}}
\def \P {{\mathbb P}}

\def \R {{\mathbb R}}
\def \N {{\mathbb N}}
\def \S {{\mathbb S}}
\def \I {{\mathbf{1}}}
\newcommand{\Ind}[1]{\mathbf{1}_{\left\{ #1 \right\}}}
\def \Pol {{\text{Pol}_n (E)}}

\newcommand{\mail}[1]{\href{mailto:#1}{\texttt{#1}}}

\title{Polynomial Diffusion Models for Life Insurance Liabilities}
\author{Francesca Biagini\footnote{Main affiliation: Department of Mathematics, LMU, Theresienstra{\ss}e, 39, 80333 Munich, Germany, Email: \mail{biagini@math.lmu.de}} \footnote{Secondary affiliation: Department of Mathematics, University of Oslo, Box 1053, Blindern, 0316, Oslo, Norway.}
\and Yinglin Zhang\footnote{Department of Mathematics, LMU, Theresienstra{\ss}e, 39, 80333 Munich, Germany, Email: \mail{zhang@math.lmu.de}}} 

\date{03 August 2016}

\begin{document}
\maketitle 

\begin{abstract}
	In this paper we study the pricing and hedging problem of a portfolio of life insurance products under the benchmark approach, where the reference market is modelled as driven by a state variable following a polynomial diffusion on a compact state space. Such a model can be used to guarantee not only the positivity of the OIS short rate and the mortality intensity, but also the possibility of approximating both pricing formula and hedging strategy of a large class of life insurance products by explicit formulas.
\end{abstract}

\




\section{Introduction}

The goal of this paper is to analyse the problem of pricing and hedging portfolios of life insurance liabilities under the new approach of combining the \emph{benchmark methodology} and the existence of a \emph{polynomial diffusion} state variable which moves on a compact state space and drives the reference market. We consider on the market OIS bonds as well as longevity bonds, both modelled as function of the state variable representing the underlying risk factors, possibly including macro-economic variables, environmental and social indicators. In this way we also introduce a dependence structure between OIS short rate and mortality intensity. Recent studies, see e.g. \cite{Liu}, \cite{Dac} and \cite{Dee-Gra}, confirm the important role of this dependence structure when pricing life insurance liabilities.

The biggest advantage of polynomial diffusions is that they give explicit and tractable formula for conditional expectation of polynomial functions of state variable, see \cite{Fil-Lar}. In our model, we focus in particular on the case when the state variable takes value in a compact state space, which is studied in detail in \cite{Lar-Pul}. Under our model assumptions, the compactness of the state space guarantees the positivity of both OIS short rate and mortality intensity, as well as the possibility of using polynomial approximation for pricing and hedging life insurance liabilities. Nevertheless, the compactness of the state variable implies also the boundedness of OIS short rate and mortality intensity. However, the assumption of bounded short rate is not new to the literature (see e.g. \cite{Sch-Som}, \cite{Fle-Hug}, \cite{Del-Shi}, \cite{Ant-Spe}, \cite{Lim-Pri} etc.), as well as the assumption of bounded mortality intensity (see e.g. \cite{Li-Szi}). This last assumption is also supported by recent statistical studies of mortality rates (see e.g. \cite{Dem} and \cite{Jar-Kry}) and can be understood in terms of confidence region (see \cite{Li-Szi}).

Unlike \cite{Fil-Lar}, we do not assume the existence of a state price density which is equivalent to the existence of a martingale measure. 
We work under the historical or real world probability measure $\P$ and assume only the existence of a benchmark or numéraire portfolio, as proposed by E. Platen (see e.g. \cite{Pla}, \cite{Pla2}, \cite{Pla3}). As pointed out by \cite{Hul-Sch}, this assumption equals the no unbounded profit with bounded risk condition, which is weaker than the classic risk-neutral condition of no free-lunch with vanishing risk (see e.g. \cite{Del-Sch}) equivalent to the existence of a martingale measure. The hybrid market composed by both financial and insurance market is intrinsically incomplete due to the presence of additional orthogonal sources of randomness given by the mortality risk. The securitization under the benchmark approach can be resolved by the so called benchmarked risk-minimizing method, see e.g. \cite{Pla4}, \cite{Pla3}, \cite{Du-Pla}  and \cite{Bia-Cre} for a single payoff.  Here we extend the results in \cite{Bia-Cre} of hedging in incomplete markets under the benchmark approach to the case of assets with dividend payments and analyse its relation with the real world pricing formula.

The paper is structured as follows.
Section \ref{sec: state variable} provides a brief introduction of polynomial diffusions. Section \ref{sec: model assumptions} gives the basic model structure assumptions for a portfolio of life insurance policies under the intensity-based approach. A benchmark portfolio is assumed to exist. Benchmark portfolio, OIS bond and longevity bond are then modelled explicitly in terms of a possibly multi-dimensional state variable which follows a polynomial diffusion on a compact state space. 
The real world pricing formula is provided explicitly in Section \ref{sec: pricing and hedging insurance products} for all three building blocks of insurance products (pure endowment, term insurance, annuity) in the case of payoff given by polynomial functions of the state variable. We then show how these explicit results can be used to approximate more general forms of payment process. For the sake of simplicity, we derive explicitly the benchmarked risk-minimizing strategy and its polynomial approximation only for pure endowment contracts, since the other two cases are similar.
In Section \ref{sec: simulation} we use a 2-dimensional state variable and calibrate our model to MSCI and LLMA index under linear specification of the inverse of benchmark and the longevity index. We show how this parsimonious specification can already produce a good fit to market data.
In Appendix \ref{app: benchmarked risk-min payment process} we extend the real world pricing formula and the benchmarked risk-minimizing method to dividend payments and discuss their relationship.

\section{Polynomial diffusion}\label{sec: state variable}

In this section, we give a synthetic overview of the most important results for polynomial diffusions which will be used to model our market. For details about these processes, see \cite{Fil-Lar}. 

Let $E \subset \R^d$ be a compact set with non-empty internal part, called state space. Let $\S^d$ denote the space of real symmetric $d \times d$ matrices, and $\S^d_+$ the convex cone of positive semidefinite symmetric matrices. For a given $n \in \mathbb{N}$, we denote furthermore by $\Pol$ the following finite-dimensional vector space
\[
	\Pol := \{  \text{polynomials on $E$ of degree $\leqslant n$} \},
\]
and by $N_n$ the dimension of $\Pol$. We consider a fixed time horizon $[0,T]$ and an $E$-valued process $Z := (Z_t)_{t \in [0, T]}$, with the following dynamics
\begin{equation}\label{eq: dynamics Z}
	\ud Z_t = b (Z_t) \ud t + \sigma (Z_t) \ud W_t,
\end{equation}
realized on a filtered probability space $(\Omega, \F, \Fb, \P)$, with $\Fb := (\F_t)_{t \in [0,T]}$. $W:=(W_t)_{t \in [0,T]}$ is a $d$-dimensional $\Fb$-Brownian motion, $\sigma : \R^d \rightarrow \R^{d \times d}$ is a continuous function, ${a} := {\sigma}{\sigma}^\top $, ${a}$ and ${b}$ are two fixed maps
\[
	{a}: \R^d \rightarrow \S^d, \ \ \ \ {b}: \R^d \rightarrow \R^d,
\]
such that
\begin{equation}\label{eq: condition a, b}
	\left. {a}_{ij} \right|_E \in \text{Pol}_2(E), \ \ \ \ \left. {b}_{i} \right|_E  \in \text{Pol}_1(E), 
\end{equation}
for all $i, j = 1, ... , d$. The initial value $Z_0$ is assumed to be constant and belonging to $E$.

We consider the following operator $\mathbf{G}$ associated to the process $Z$ 
\[
	\mathbf{G}: \mathcal{C}^2(\R^d) \longrightarrow \R,
\]
defined by
\begin{equation}\label{operatore g}
	\mathbf{G} f(z) := \frac{1}{2} \text{Tr}({a}(z) \nabla^2 f(z)) + {b}(z)^\top \nabla f(z), \ \ \ \ z \in \R^d.
\end{equation}

\begin{defn}
	An $E$-valued process $Z$ satisfying (\ref{eq: dynamics Z}) is called a \emph{polynomial diffusion on $E$} if 
	\[
		\mathbf{G} \Pol \subseteq \Pol , \ \ \ \ \text{for all } \ \ n \in \N.
	\]
\end{defn}

Lemma 2.2 of \cite{Fil-Lar} shows that $Z$ is a polynomial diffusion on $E$ if and only if (\ref{eq: condition a, b}) holds. In such case, for every fixed $n \in \N$, its associated operator $\mathbf{G}$ admits a unique matrix representation $G_n \in \R^{N_n \times N_n}$ restricted to $\Pol$. That is, for each $p \in \Pol$ with coordinate representation
\begin{equation}\label{eq: coordinate representation p}
	p(z) = H_n(z)^\top \vec{p}, \ \ \ \ z \in \R^d,
\end{equation}
where $H_n(z)$ is a fixed basis vector of $\Pol$ and $\vec{p} \in \R^{N_n}$, we have
\[
	\mathbf{G} p(z) = H_n(z)^\top G_n \vec{p}, \ \ \ \ z \in \R^d.
\]

The following proposition gives one of the most important results for polynomial diffusions: the $ \Fb$-martingale generated by a polynomial function of $Z_T$ is again given by a polynomial function of the state variable with deterministic time-dependent coefficients.
\begin{prop}[Polynomial conditional expectation]\label{prop pol cond exp}
	Let $n \in \N$, $p \in \Pol$ with coordinate representation (\ref{eq: coordinate representation p}) and $0\leqslant t \leqslant T$. If $Z$ is a $E$-solution to (\ref{eq: dynamics Z}), then
	\begin{equation}\label{eq: pol cond exp}
		\mathbb{ E} \left [ \left.  p(Z_T) \right|  \F_t \right ] = H_n(Z_t)^\top e^{(T-t)G_n} \vec{p}. 
	\end{equation}
\end{prop}

\begin{proof}
	See Theorem 3.1 of \cite{Fil-Lar}.
\end{proof}

	For the sake of simplicity, we introduce the notation
	\[
		\hat p_{(t,T)} (Z_t) := H_n(Z_t)^\top e^{(T-t)G_n} \vec{p}, \ \ \ \ t \in [0,T].
	\]

The following two theorems give some sufficient conditions on the state space $E$, under which an $E$-valued state variable $Z$ admits weak uniqueness and existence.

\begin{theorem}\label{theorem weak uniqueness}
	Let $Z$ be an $E$-valued solution to (\ref{eq: dynamics Z}). If $E$ is compact, then $Z$ is weakly unique, i.e. any other $E$-valued solution to (\ref{eq: dynamics Z}) with initial value $Z_0$ has the same law as $Z$.
\end{theorem}

\begin{proof}
	See Theorem 5.1 of \cite{Fil-Lar}.
\end{proof}

\begin{rem}
	For $d=1$, we have also strong uniqueness (or pathwise uniqueness) for a $E$-valued solution to (\ref{eq: dynamics Z}), see \cite{Rog-Wil}. Further discussion about the strong uniqueness for a  generic dimension is provided in \cite{Lar-Pul}.
\end{rem}

\begin{theorem}\label{theorem existence}
	Let $\mathcal{P}$ be a family of polynomials on $\R^d$. If the boundary of the state space $E$ is defined by these polynomials, i.e.
	\[
		E = \{ x \in \R^d : p(x) \leqslant 0 \text{ for all } p \in \mathcal{P} \},
	\]
	then the following conditions on the parameters $a$ and $b$ guarantee the existence of an $E$-valued solution to (\ref{eq: dynamics Z}):
	\begin{enumerate}
		\item $ a \in \S^d_+ $;
		\item $ a \nabla  p = 0$ on $\{  p = 0 \}$ for each $ p \in \mathcal P$;
		\item $\mathbf{G}  p > 0$ on $E \cap \{ p = 0 \}$ for each $\ p \in \mathcal P$.
	\end{enumerate}
\end{theorem}

\begin{proof}
	This theorem is a simplified version of Theorem 5.3 of \cite{Fil-Lar}.
\end{proof}

\begin{rem}
	We stress that every $E$-solution to (\ref{eq: dynamics Z}) is a strong Markov process thanks to the compactness of the state space and the weak uniqueness given by Theorem \ref{theorem weak uniqueness} (see e.g. Theorem 4.6 of \cite{Dur}). The operator defined in (\ref{operatore g}) is thus the extended Markov generator of $Z$.
\end{rem}


\section{The Setting}\label{sec: model assumptions}

In this paper, we follow the intensity-based reduced-form approach (e.g. \cite{Bie-Rut}) in order to describe life insurance derivatives linked to a portfolio of life insurance contracts.
	
	We fix a finite time horizon $[0,T]$ with $0 < T < \infty$ and a filtered probability space $(\Omega, \G, \Gb , \P)$, with $\Gb := (\G_t)_{t \in [0, T]}$, $\G=\G_{T},  \G_0 \ \text{trivial}$.
	The filtration $\Gb$ represents all the information flow available to the insurance company.
	$\P$ is interpreted as the historical or real world probability measure.	
	We consider a life insurance portfolio with $n$ policyholders ($n \in \N$). For $i = 1, ..., n$, the decease time of $i$-th insured person is assumed to be a $\G$-measurable random time $\tau^i: (\Omega, \G, \P) \rightarrow \R_+$, , such that $\tau^i > 0$ a.s., $\mathbb{P}(\tau^i = 0) = 0$ and $\P(\tau^i > t ) > 0$ for any $t \in [0,T]$. These last conditions mean that the decease of a policyholder does not occur in the first instant and may never occur during the considered time horizon $[0, T]$. We denote the death counting process by $N := (N_t)_{t \in [0,T]}$,
	\begin{equation}\label{eq: death counting process}
		N_t = \sum_{i = 1}^n \Ind{\tau^i \leqslant t}, \ \ \ \ t \in [0,T].
	\end{equation}
	Furthermore, we assume $\Gb = \Fb \vee \Hb^1 \vee ... \vee \Hb^n$ where for $i = 1,...,n$, $\Hb^i := (\H^i_t)_{t \in [0,T]}$, $\Fb := (\F_t)_{t \in [0,T]}$, $\F_0 = \H^i_0 = \G_0  = \{ \emptyset, \Omega \}$, $\G = \G_{T} = \F_{T} \vee \H_T^1 \vee ... \vee \H_T^n$. The filtration $\Fb$ is called reference filtration. For each $i = 1,...,n$, the filtration $\Hb^i$ is generated by the jump process of $\tau^i$ defined by $H^i:= (H^i_t)_{t \in [0, T]}$ with $H^i_t = \mathbf{1} _{ \{\tau^i \leqslant t\} }$ for every $t \in [0,T]$, i.e. the filtration $\Hb^i$ represents the information through time relative to $i$-th policyholder's life status described by $\tau^i$. 
	All the filtrations are assumed to be complete and right-continuous. Every random variable $\tau^i$ for $i = 1,...,n$ is clearly a $\Gb$-stopping time, but not necessary an $\Fb$-stopping time.
 	We assume also that for each $i, j = 1, ... , n$, $\mathbb{P}\left( \tau^i = \tau^j \right) = 0$ and that $\tau^i$ and $\tau^j$ are $ \Fb$-conditionally independent, i.e. for any $t \in [0,T]$ and $s, r \in [0,t]$,
 	\[
 		\P \left. \left( \tau^i > s, \tau^j > r \right| \F_t \right) = \P \left. \left( \tau^i > s \right| \F_t \right) \P \left. \left( \tau^j > r \right| \F_t \right).
 	\]
 	This assumption can be interpreted in the following way. The life status of policyholders is influenced by some common systematic factors related to environmental, economic, social and financial conditions captured by the information flow represented by $\Fb$, but every policyholder has also his own purely individual or idiosyncratic decease factors (e.g. illness, accident etc.).
 	\\
 	Moreover, we assume that 
	\begin{enumerate}
		\item The subfiltration $\Fb$ of $\Gb$ satisfies the so called $H$-hypothesis, i.e. every $\Fb$-martingale is also a $\Gb$-martingale. 
		
		\item For every $i = 1,...,n$, the following non-negative $\Fb$-submartingale $F^i := (F^i_t)_{t \in [0,T]}$, called conditional cumulative distribution function of $\tau^i$
			\[
				F^i_t := \P \left. \left( \tau^i \leqslant t \right| \F_t \right), \ \ \ \ t \in [0, T]
			\]
			satisfies
			\[
				F^i_t < 1, \ \ \ \ t \in [0, T],
			\]
			and is absolutely continuous with respect to Lebesgue's measure.
	\end{enumerate}
	For every $i=1,...,n$, the hazard process $\Gamma^i := (\Gamma^i_t)_{t \in [0,T]}$ of $\tau^i$ 
	\[
		\Gamma^i_t := - \ln \left(1 - F^i_t \right), \ \ \ \ t \in [0, T]
	\]
	is well defined, absolutely continuous and increasing (see Lemma 6.1.2 of \cite{Bie-Rut}). On one hand, this implies that, for every $i=1,...,n$, $\tau^i$ avoids all $\Fb$-stopping times and is a $\Gb$-totally inaccessible stopping time (see \cite{Del-Mey}).
	On the other hand, the absolute continuity of $\Gamma^i$ implies the existence of a non negative $\Fb$-progressively measurable mortality intensity $\mu^i : = (\mu_t)_{t \in [0, T]}$ with integrable sample paths such that
	\[
		\Gamma^i_t = \int_0^t \mu^i_u \ud u, \ \ \ \ t \in [0, T].
	\]
	We can take a $\Fb$-predictable version of $\mu^i$, see Lemma 1.36 of \cite{Jac}.	
	For the sake of simplicity, we assume a homogeneous insurance portfolio where all policyholders have the same mortality rate and we denote the common conditional cumulative distribution function, hazard process and mortality intensity respectively by $F := (F_t)_{t \in [0.T]}$, $\Gamma := (\Gamma_t)_{t \in [0,T]}$ and $\mu := (\mu_t)_{t \in [0,T]}$, i.e. $F = F^i$, $\Gamma = \Gamma^i$, $\mu = \mu^i$ for all $i = 1,...,n$. This happens for example when all policyholders belong to the same age cohort in the same country\footnote{The number of the reference population group is normally much bigger than $n$.}. We assume furthermore that there is a publicly accessible index based on mortality data of the given age group. According to \cite{Cai-Bla}, such index is modelled by the process $I := (I_t)_{t \in [0, T]}$, where
	\[
		I_t := 1 - F_t = e^{-\Gamma_t} = e^{- \int^t_0 \mu_u \ud u}, \ \ \ \ t \in [0,T],
	\]
	and is called survival index (or longevity index or longevity process) related to the given group of people\footnote{The assumption of a homogeneous portfolio can be easily relaxed and extended to the case of more population groups, subdivided by age cohorts, countries etc.. In that case, every group of people is assumed to have its correspondent longevity index and mortality intensity.}.
	We stress that under our assumptions, the reference filtration gives information about the (common) mortality intensity $\mu$, or equivalently the longevity index $I$, but not about a single decease event itself.
	For $i = 1,...,n$, we introduce the process $L^i := (L^i_t)_{t \in [0,T]}$ associated to the $i$-th policyholder
	\begin{equation}\label{eq: process L}
		L^i_t := (1 - H^i_t) e^{\Gamma_t} = \Ind{\tau^i > t} e^{\int^t_0 \mu_s \ud s}, \ \ \ \ t \in [0, T].
	\end{equation}
	Proposition 5.7 and Proposition 5.8 of \cite{Bar} show that $L^i$ is a $\Gb$-martingale and satisfies
	\[
		L^i_t = 1 - \int_{]0,t]} L^i_{u-} \ud M^i_u, \ \ \ \ t \in [0, T],
	\]
	where $M^i := (M^i_t)_{t \in [0,T]}$ is a $\Gb$-martingale defined by
	\[
		M_t^i = H^i_t - \Gamma^i_{t \wedge \tau^i}, \ \ \ \ t \in [0, T].
	\]
	Hence, $(\Gamma_{t \wedge \tau^i})_{t \in [0,T]}$ is the $\Gb$-compensator of $H^i$. In other words, the hazard process $\Gamma^i$ coincides up to $\tau^i$ with the so called $(\Fb, \Gb)$-martingale hazard process $\Lambda^i := \left(\Lambda^i_t\right)_{t \in [0,T]}$ of $\tau^i$.
	Under the assumption of a homogeneous insurance portfolio, if we define the $\Gb$-martingale $M := (M_t)_{t \in [0,T]}$ as 
	\[
		M := \sum_{i=1}^n M^i,
 	\]
 	then
 	\begin{equation}\label{eq: process M}
 		M_t = N_t - (n - N_{t-}) \Gamma_t, \ \ \ \ t \in [0,T], 
 	\end{equation}
 	i.e. the $\Gb$-compensator of the death counting process $N$ is given by $\left((n - N_{t-}) \Gamma_t\right)_{t \in [0,T]}$.
	 
\subsection{Market model}

We assume that our reference market is frictionless and there are $l$ liquidly traded primary assets with price processes $S^i:=(S^i_t)_{t \in [0, T]}$, $i = 1, ... , l$, following real-valued continuous $(\P, \Fb)$-semimartingales. The asset vector is denoted by $S:=(S^i)_{i =1,...,l}$. For now we do not specify the nature and the dynamics of these primary assets.
\begin{rem}
	We remark that primary assets are here assumed to be $\Fb$-adapted. Unlike in the credit risk context where products associated to a single default (e.g. bankruptcy of an important finance institution) are traded on the financial market over time, this is not the case of life insurance products based on the decease of single persons. However, primary assets can include some mortality/longevity index-linked securities, like longevity bond.
\end{rem}

\begin{defn}
	A \emph{trading strategy} is a $\R^l$-valued $\Gb$-predictable $S$-integrable process $\delta := (\delta_t)_{t \in [0, T]}$.
\end{defn}
The space of all $\R^l$-valued $\Gb$-predictable $S$-integrable processes is indicated by $L(S,\P,\Gb)$. 

\begin{defn}
	A \emph{portfolio} or \emph{value process} $S^\delta:=(S^\delta_t)_{t \in [0,T]}$ associated to a trading strategy $\delta$ is defined by the following càdlàg optional process\footnote{We follow the definition given in \cite{Bia-Pra}.}
	\[
		S^\delta_{t-} = \delta_t^\top S_t = \sum_{i=1}^l \delta^i_t S^i_t, \ \ \ \ t \in [0,T].
	\]
	It is called \emph{self-financing} if 
	\[
		S^\delta_t = S_0^\delta + \int_0^t \delta_u^\top \ud S_u = S_0^\delta + \sum_{i=1}^l \int_0^t \delta^i_u \ud S^i_u, \ \ \ \ t \in [0.T].
	\]
\end{defn}

We note that according to our definition, for $t \in [0, T]$ and $i = 1,...,l$, the variable $\delta^i_t$ represents the amount of $i$-th primary asset held at time $t$. We define the following set
\[
	\mathcal{V}^{+}_x = \{ S^\delta \text{ self-financing }| \ \delta \in L(S, \P, \Gb), \ S_0^\delta = x > 0, \ S^\delta > 0\}.
\]

\begin{defn}\label{def: benchmark}
	A \emph{benchmark} or \emph{numéraire portfolio} $S^{\delta^*}:=(S^{\delta^*}_t)_{t \in [0,T]}$ is a portfolio in the set $\mathcal{V}^{+}_1$, such that every portfolio $S^\delta \in \mathcal{V}^{+}_1$ when discounted by $S^{\delta^*}$ forms a $(\P, \Gb)$-supermartingale, i.e. 
	\[
		\frac{S^\delta_t}{S^{\delta^*}_t} \geqslant \condespg{\frac{S^\delta_s}{S^{\delta^*}_s}}, \ \ \ \ t,s \in [0,T], \ \ s \geqslant t.
	\]
\end{defn}

In our framework, we assume only the existence of a numéraire portfolio denoted by $S^*$ and not necessarily the existence of an equivalent martingale measure.
We note that according to our definition of benchmark (or numéraire) portfolio, we exclude that the benchmark may contain claims related to single deceases. 

\begin{defn}
	We call \emph{benchmarked value} the value of any security or portfolio $S^\delta$ when discounted by the benchmark portfolio and we denote it by process $\tilde S^\delta$, i.e.
	\[
		\tilde S^\delta := \frac{S^\delta}{S^*}.
	\]
\end{defn}

Since in our framework all primary assets are assumed to be continuous, the following result holds.

\begin{prop}\label{prop: benchmarked primary assets}
	The benchmarked vector process of primary assets $\tilde S := \frac{S}{S^*}$ is a $(\P, \Gb)$-local martingale.
\end{prop}

\begin{proof}
	See e.g. \cite{Bia-Cre}.
\end{proof}

The cash flow received by the policyholder from the insurer over time can be seen as a dividend payment which is usually modelled by a process $D := ( D_t)_{t \in [0, T]}$ of finite variation or more in general a $(\P, \Gb)$-semimartingale. We denote by $A:=(A_t)_{t \in [0,T]}$ the benchmarked value of the cumulative liabilities of the insurer towards a policyholder, formally
\begin{equation}\label{eq: A process definition}
	A_t := \int^t_0 (S^*_u)^{-1} \ud D_u, \ \ \ \ t \in [0,T].
\end{equation}We assume that $D$ is defined such that $A$ is square integrable\footnote{A process $X :=(X_t)_{t \in [0,T]}$ is square integrable if $\sup_{t \in [0,T]} \esp{X_t^2} < \infty$.}.

\begin{defn}
	The following formula for a dividend process, which settles at time $T$, is called \emph{real world pricing formula},
	\begin{equation}\label{eq: pricing formula dividend}
		V_t := S^*_t \condespg{A_T - A_t} = S^*_t \condespg{\int_{]t, T]} (S^*_u)^{-1} \ud D_u},
	\end{equation}
	for $t \in [0, T]$.
\end{defn}
This definition is a generalization of the so called ex-dividend price process (see e.g. \cite{Bac-Bie} and \cite{Jam}) which gives the current value of all future cash flow the insurer has to pay in the risk-neutral valuation context, i.e. when a martingale measure is assumed to exist.
In the general case, if the existence of a martingale measure is not assumed, the benchmarked value of the price process $V := (V_t)_{t \in [0, T]}$ corresponds to the \emph{benchmarked risk-minimizing price} as we explain in Appendix \ref{app: benchmarked risk-min payment process}.

It is well known that the hybrid market extended with the introduction of insurance contracts is intrinsically incomplete even when the reference market is complete. This fact is due to the additional source of randomness given by mortality risk and the lack of liquidity in trading life insurance products based on single decease over time, which create an unhedgeable basis risk (see e.g. \cite{Bia-Rhe-r}). Hence in our setting we need to choose a hedging method for incomplete markets.
The so called risk-minimizing method, which aims to find the optimal replicating strategy and minimizing the expected quadratic risk,  appears to be a natural approach when market incompleteness is due to external source of randomness as discussed in \cite{Bia}. It was originally introduced in \cite{Fol-Son} for a single payoff and then extended to dividend processes and applied to insurance contracts in e.g. \cite{Mol-u}, \cite{Mol-i}, \cite{Dah-Mol}, \cite{Bar}, \cite{Bia-Sch} and \cite{Bia-Rhe-r}. However, the results in these papers require a risk-neutral contest. The relationship of benchmark approach and risk-minimization has been analysed in \cite{Pla3} and \cite{Bia-Cre} for a single payoff. In Appendix \ref{app: benchmarked risk-min payment process} we show how this method can be easily reviewed to include dividend payments.

We now consider a state variable process $Z$ representing the underlying risk factors, possibly including macro-economic variables, environmental and social indicators. Let $m_1, m_2 \in \N$ be such that $m_1 + m_2 = d$. Furthermore, we assume $Z$ is given by a polynomial diffusion of the form
\[
 	Z = \left( \begin{array}{c} X \\ Y \end{array} \right), 
\]
with
\begin{equation}\label{eq: dynam (X,Y)}
	\left\{\begin{array}{rcl}
		\ud X_t &= &b(X_t)\ud t + \sigma ( X_t)\ud W_t, \\
		\ud Y_t &= &\overline b(X_t, Y_t)\ud t,
	\end{array} \right.
\end{equation}
for $t \in [0,T]$, realized on the filtered probability space $(\Omega, \F, \Fb, \P)$, where $W:=(W_t)_{t \in [0,T]}$ is a $d$-dimensional $\Fb$-Brownian motion, $\sigma: \R^{m_1} \rightarrow \R^{m_1 \times d}$ is a continuous function, $a := \sigma \sigma^\top$,
\[
\begin{array}{ccc}
	a: \R^{m_1} \longrightarrow \S^{m_1} , &b: \R^{m_1} \longrightarrow \R^{m_1}  , &\overline b : \R^{d} \longrightarrow \R^{m_2}, \\
	a_{ij} \in \text{Pol}_2(\R^{m_1})  ,  &b_i \in \text{Pol}_1 (\R^{m_1})  , &\overline b_k \in \text{Pol}_1 (\R^{d}),
\end{array}
\]
for all $i,j = 1,...,m_1 $ and $k = 1, ... , m_2$. 
The (compact) state space $E$ is given by
\[
	E = E^X \times E^Y \ , \ E^X \subseteq \R^{m_1} \text{ and } \ E^Y \subseteq \R^{m_2} \ , 
\]
where $E^X$ is the state space of process $X$ and $E^Y$ the state space of $Y$, respectively.

The following proposition shows that if $E^X$ is a compact set of $\R^{m_1}$, then $E^Y$ is compact as well.
\begin{prop}\label{prop: boundness Y}
	If there is a constant $C$ such that for every $t \in [0, T]$,
	\[
		\norm{X_t} \leqslant C,
	\]
	then $Y$ is uniformly bounded.
\end{prop}

\begin{proof}
	Since $\bar b$ is a linear function of $Z$, the $Y$-dynamics can be written as
	\[
		\ud Y_t = (A X_t + B Y_t + c) \ud t,
	\]
	with $A \in \R^{m_2 \times m_1}$, $B \in \R^{m_2 \times m_2}$, $c \in \R^{m_2}$. In particular, we have
	\[
		\norm{Y_t} \leqslant b \int^t_0 \norm{Y_u} \ud u + a C t + \norm{c}t + \norm{Y_0},
	\] 
	where $a$ and $b$ are respectively some matrix norms of $A$ and $B$.\\	
	The Gr\"{o}nwall's inequality yields
	\[	
		\norm{Y_t} \leqslant \bar C  + b \int^{t}_o e^{b(t - s)} ds \leqslant \bar C  + b \int^{T}_o e^{b(T - s)} ds < \infty,
	\]
	for all $t \in [0,T]$, with $\bar C$ a suitable constant, i.e. $Y$ is uniformly bounded.
\end{proof}

For the sake of simplicity, in the rest of this paper the degree of a generic polynomial function $p$ will be indicated by $\bar p$.


We model the benchmark portfolio as driven by the state variable $Z$ in the following way
\begin{equation}\label{eq: benchmark process}
	\frac{1}{S^*_t} := e^{ - \alpha t} p(Z_t), \ \ \ \alpha \in \R  \text{ and } p \text{ strictly positive polynomial on } E,
\end{equation}
for every $t \in [0, T]$. We note that according to our definition (\ref{eq: benchmark process}), the benchmark portfolio is $\Fb$-adapted and continuous. In \cite{Fil-Lar} a similar dynamics is specified for a state price density, here we choose to model the benchmark portfolio.

Now we assume that both risk-free OIS bond and longevity bond maturing at $T$ are among primary assets and we indicate their value processes respectively by $(P(t,T))_{t \in [0,T]}$ and $(P^l(t, T))_{t \in [0,T]}$. In view of Proposition \ref{prop: benchmarked primary assets}, we assume that their benchmarked value processes are continuous $(\P,\Gb)$-true martingales.
We stress that for compact state space $E$, the restricted polynomial $\left. p \right| _E$ admits a strictly positive minimum value, i.e.  there exists a strictly positive number $\varepsilon$ such that
\begin{equation}\label{eq: bound p}
	E \subseteq \{ z \in \R^d : p(z) \geqslant \varepsilon\}.
\end{equation}
As we will see below, under our model assumptions the condition (\ref{eq: bound p}) ensures the continuity of both risk-free OIS bond and longevity bond, as well as the non-negativity of the risk-free short rate by adjusting the parameter $\alpha$. 

A risk-free OIS bond maturing in $T$ is by definition a zero-coupon bond with unit payment at term of contract $T$, whose value at $t \in [0,T]$ is represented by
\[
	P(t,T) = S^*_t \condespg{\frac{1}{S^*_T}},
\]
which can be explicitly calculated using (\ref{eq: benchmark process}) and Proposition \ref{prop pol cond exp},
\begin{align*}
	S^*_t \condespg{\frac{1}{S^*_T}} &= e^{-\alpha(T - t)} \frac{\condespg{p(Z_T)}}{p(Z_t)} = e^{-\alpha(T - t)} \frac{\condespf{p(Z_T)}}{p(Z_t)}\\
			&=e^{-\alpha(T - t)} \frac{H_{\bar p}(Z_t)^\top e^{(T - t)G_{\bar p}} \vec p}{p(Z_t)} = e^{-\alpha(T - t)} \frac{\hat p_{(t,T)} (Z_t)}{p(Z_t)},
\end{align*}
where second equality follows by Lemma 6.1.1 of \cite{Bie-Rut}.
Due to (\ref{eq: bound p}), the process above is well-defined and continuous. Hence we have
\begin{equation}\label{eq: OIS bond definition}
	P(t,T) = e^{-\alpha(T - t)} \frac{\hat p_{(t,T)} (Z_t)}{p(Z_t)}, \ \ \ \ t \in [0,T].
\end{equation}
We note that in our model the value of a OIS $T$-bond in time $t$ is a ratio of polynomial functions of $Z_t$, with explicit deterministic time-dependent coefficients.

The risk-free short rate process $r := (r_t)_{t \in [0, T]}$ can be explicitly calculated from the risk-free OIS bond dynamics (\ref{eq: OIS bond definition}), for $t \in [t,T]$
\begin{align*}
	r_t :&= - \left. \partial_T \log P(t, T) \right|_{T=t} \\
		&= - \left. \partial_T \log \left( e^{-\alpha(T - t)} \frac{\hat p_{(t,T)} (Z_t)}{p(Z_t)} \right) \right|_{T=t}\\
		&= - \left. \partial_T \log \left( e^{-\alpha(T - t)} \right) \right|_{T=t} - \left. \partial_T \log \left( \frac{H_{\bar p}(Z_t)^\top e^{(T-t)G_{\bar p}} \vec{p}}{p(Z_t)} \right) \right|_{T=t}\\
		&= \alpha - \left. \left[ \frac{p(Z_t)}{H_{\bar p}(Z_t)^\top e^{(T-t)G_{\bar p}} \vec{p}} \ \partial_T \left( \frac{H_{\bar p}(Z_t)^\top e^{(T-t)G_{\bar p}} \vec{p}}{p(Z_t)} \right) \right] \right|_{T=t}\\
		&= \alpha - \left. \left[ \frac{p(Z_t)}{H_{\bar p}(Z_t)^\top e^{(T-t)G_{\bar p}} \vec{p}} \frac{H_{\bar p}(Z_t)^\top G_{\bar p} e^{(T-t)G_{\bar p}} \vec{p}}{p(Z_t)}  \right] \right|_{T=t}\\
		&= \alpha - \frac{p(Z_t)}{H_{\bar p}(Z_t)^\top  \vec{p}} \frac{H_{\bar p}(Z_t)^\top G_{\bar p} \vec{p}}{p(Z_t)} \\
		&= \alpha - \frac{H_{\bar p}(Z_t)^\top G_{\bar p} \vec p}{p(Z_t)},
\end{align*}
since $e^{(T-t)G_{\bar p}} = 1$ when $T=t$.
In particular, the compactness of the state space $E$ and (\ref{eq: bound p}) provide that
\[
	\frac{H_{\bar p}(Z_t)^\top G_{\bar p} \vec p}{p(Z_t)}
\]
has an upper bound $\overline \alpha$ and a lower bound $\underline{\alpha}$ uniformly in $t \in [0, T]$.
By choosing $\alpha = \overline \alpha$, the short rate takes positive value in $[ 0 , \overline \alpha - \underline{\alpha} ]$.

Following the definition in \cite{Cai-Bla} and \cite{Bla-Cai}, a longevity bond maturing at $T$ is a index-linked zero-coupon bond with final payment at $T$ equal to the value of a given survival index at $T$. 
Unlike the usual intensity-based approach, we model first the survival index and then derive the mortality intensity dynamics. 
We use the $Y$-component of the state variable $Z$ to model the survival index
\begin{equation}\label{eq: longevity index}
	I_t := e^{- \gamma t} q(Y_t), \ \ \ \gamma \in \R \text{ and } q \text{ strictly positive polynomial on } E^Y,
\end{equation}
for every $t \in [0,T]$.
The parameter $\gamma$ will be used to adjust the value level of mortality intensity.
The same argument as before leads to the existence of a strictly positive number $\delta$ such that
\begin{equation}\label{eq: bound q}
	E^Y \subseteq \{ y \in \R^{m_2} : q(y) \geqslant \delta\}.
\end{equation}

The formula for the mortality intensity $(\mu):= (\mu_t)_{t \in [0,T]}$ can be obtained immediately 
\begin{equation}\label{eq: mortality intensity}
	\mu_t = - \left.\partial_T \log \left(I_T \right)  \right|_{T=t} = \gamma - \frac{\nabla q(Y_t)^\top \overline b(X_t, Y_t)}{q(Y_t)},
\end{equation}
for all $t \in [0.T]$.
Analogously to the case of risk-free short rate $r$, thanks to the compactness of $E^Y$ and the condition (\ref{eq: bound q}), we have that uniformly in $t$ the quantity
\[
	\frac{\nabla q(Y_t)^\top \overline b(X_t, Y_t)}{q(Y_t)}
\]
has an upper bound $\overline \gamma$ and a lower bound $\underline{\gamma}$. If we set $\gamma = \overline \gamma$, the mortality intensity has then a positive value range $[ 0 , \overline \gamma - \underline{\gamma} ]$.

Similar to the OIS bond case, under our assumption and by using definition (\ref{eq: longevity index}) of survival index and Proposition \ref{prop pol cond exp}, we can calculate explicitly the value of a $T$-longevity bond $P^l(t,T)$ at $t \in [0,T]$,
\begin{equation}\label{eq: def longevity bond}
\begin{aligned}
	P^l(t, T) &=S^*_t \condespg{\frac{I_T}{S^*_T}} = S^*_t \condespg{(S_T^{*})^{-1} e^{- \gamma_T} q(Y_T)}\\
				   &= I_t e^{- (\alpha + \gamma) (T-t)} \frac{\condespf {p(Z_T) q(Y_T)}}{p(Z_t) q(Y_t)}\\
				   &= I_t e^{- (\alpha + \gamma) (T-t)} \frac{\hat {pq}_{(t,T)}(Z_t)}{p(Z_t) q(Y_t)}\\
				   &= e^{ - \gamma T - \alpha (T-t)} \frac{\hat {pq}_{(t,T)}(Z_t)}{p(Z_t)},
\end{aligned}
\end{equation}
where in the second equality we use Lemma 6.1.1 of \cite{Bie-Rut}. Condition (\ref{eq: bound p}) guarantees the continuity of the process above.
Similar to the case of risk-free OIS bond $P(t,T)$, the value of longevity bond $P^l(t,T)$ at time $t$ is also a polynomial rational function of the state variable $Z_t$ with deterministic time-dependent coefficients. 


	
	

\section{Pricing and hedging life insurance liabilities}\label{sec: pricing and hedging insurance products}

It is well known that most of life insurance liabilities can be modelled as a combination of the following three building blocks, which are particular cases of dividends:

\begin{itemize}
	\item Pure endowment contract: the insurer pays only if the policyholder survives until the maturity $T$ of the contract.
	
	\item Term insurance contract: the payment is given only when the decease of the insured person occurs before or at $T$.
	
	\item Annuity contract: continuous cash flow is paid as long as the policyholder is alive or the contract is valid.
\end{itemize}

In the following sections, we compute the real world pricing formula and the benchmarked risk-minimizing strategy for the three building blocks, and show how the property in Proposition \ref{prop pol cond exp} gives explicit formulas in the case of polynomial payments, as well as a good approximation in the case of continuous payments. All theorems and notations are provided in Appendix \ref{app: benchmarked risk-min payment process}. 
For the sake of simplicity, we assume to invest only in the OIS bond and the longevity bond, i.e.
\[
	\tilde S_t = \left( \frac{P(t, T)}{S^*_t} , \frac{P^l(t, T)}{S^*_t} \right)^\top, \ \ \ \ t \in [0,T],
\]
and calculate the explicit benchmarked risk-minimizing strategy only for pure endowment contract assuming that $\dim W = n = 2$. The case of the other two building blocks is similar.

The following lemma will be used frequently.
\begin{lemma}\label{lemma: decomposition hat p}
	Let $p$ be a polynomial in $\Pol$ with coordinate representation 
	\[
		p(z) = H_n(z)^\top \vec{p},
	\]
	for $z \in E$, then for $0 \leqslant t \leqslant T$,
	\[
		\hat p_{(t,T)}(Z_t) = \condespg{p(Z_T)} = H_n(Z_0)^\top e^{T G_n} \vec p + \int^t_0 \nabla \hat p_{(u,T)}(Z_u)^\top \sigma (Z_u) \ud W_u.
	\]
\end{lemma}

\begin{proof}
	\begin{align*}
		 &\hat p_{(t,T)}(Z_t)\\
		=& \hat p_{(0,T)}(Z_0) + \int^t_0 \left( \frac{\partial}{\partial u} \hat p_{(u,T)}(Z_u) \right) \ud u + \int^t_0 \mathbf{G} \hat p_{(u,T)}(Z_u) \ud u\\ 
		 &+ \int^t_0 \nabla \hat p_{(u,T)}(Z_u)^\top \sigma (Z_u) \ud W_u\\
		=& H_n(Z_0)^\top e^{T G_n} \vec p - \int^t_0 \left( H_n(Z_u)^\top G_n e^{(T - u) G_n} \vec p \right) \ud u + \int^t_0 \left( H_n(Z_u)^\top G_n e^{(T - u) G_n} \vec p \right) \ud u\\ 
		 &+ \int^t_0 \nabla \hat p_{(u,T)}(Z_u)^\top \sigma (Z_u) \ud W_u\\
		=& H_n(Z_0)^\top e^{T G_n} \vec p + \int^t_0 \nabla \hat p_{(u,T)}(Z_u)^\top \sigma (Z_u) \ud W_u.
	\end{align*}
	The first equality is given by the It\^{o}'s formula and the second one is due to Proposition \ref{prop pol cond exp}.
\end{proof}

\subsection{Pure endowment}\label{sec: pure endowment}

A pure endowment contract provides a payment at the term $T$ of contract if the insured person is still alive. For $i= 1,...,n$, its payoff at $T$ associated to $i$-th policyholder is given by
\[
	\Ind{\tau^i > T} g_T,
\]
where the value $g_T$ is assumed to be a $\F_T$-measurable and square integrable random variable. For a homogeneous portfolio of $n$ policyholders with the same payoff $g_T$ we have
\[
	\sum_{i=1}^n \Ind{\tau^i > T} g_T = (n - N_T) g_T.
\] 
The benchmarked cumulative payment $A$ is given by
\[
	A_t =  \sum_{i=1}^n(S^*_T)^{-1} \I_{ \{\tau^i > T \}} g_T \Ind{t = T} = (S^*_T)^{-1} (n - N_T) g_T \Ind{t = T},
\]
for $t \in [0,T]$.

Let $V^T := (V^T_t)_{t \in [0,T]}$ denote the price process given by the real world pricing formula (\ref{eq: pricing formula dividend}) associated to a homogeneous portfolio of pure endowments. Under our model assumptions of Section \ref{sec: model assumptions}, we have at time $t \in [0,T]$,
\begin{align*}
	V^T_t =  & S^*_t \condespg{(S^*_T)^{-1} \sum_{i=1}^n \I_{ \{\tau^i > T \}} g_T} \\
	= & \sum_{i=1}^n S^*_t \condespg{(S^*_T)^{-1} \I_{ \{\tau^i > T \}} g_T }\\
	= & \sum_{i=1}^n \Ind{\tau^i > t} S^*_t \condespf{(S^*_T)^{-1} e^{- \int^T_t \mu_u \ud u} g_T }\\
	= & (n - N_t) e^{- (\gamma + \alpha) ( T - t)} \frac{\condespf{p(Z_T) q(Y_T) g_T }}{p(Z_t) q(Y_t)},
\end{align*}
where in the third equality we use Proposition 5.5 of \cite{Bar} combined with Corollary 5.1.1 of \cite{Bie-Rut}.
Then the benchmarked value process $\tilde S^{\bar \delta^T} := (\tilde S^{\bar \delta^T}_t)_{t \in [0,T]}$ associated to the benchmarked risk-minimizing strategy $\bar \delta^T = (\bar \delta^T_t)_{t \in [0,T]}$ of the given portfolio is
\[
	\tilde S^{\bar \delta^T}_t = (S^*_t)^{-1} V^T_t = (n - N_t) e^{- \alpha T - \gamma (T-t)} \frac{\condespf{p(Z_T) q(Y_T) g_T }}{q(Y_t)},
\]
for $t \in [0,T]$.
Proposition 5.11 of \cite{Bar} can be easily adapted to our case and together with (\ref{eq: benchmark process}) it shows that the benchmarked risk-minimizing strategy is given by $\bar \delta^T$ with
\begin{equation}
	\bar \delta^T_t = (n - N_{t-}) e^{- \alpha T - \gamma (T-t)} q^{-1}(Y_t) \phi_t,
\end{equation}
for $t \in [0,T]$, where the vector process $\phi := (\phi_t)_{t \in [0, T]}$ is obtained by the Galtchouk-Kunita-Watanabe decomposition of $U_t := \condespf{p(Z_T) q(Y_T) g_T }$
\begin{equation}\label{eq: gkw decomposition of U}
	U_t =\esp{p(Z_T) q(Y_T) g_T} + \int_0^t \phi_u^\top \ud \tilde S_u + L_t^{U}, \ \ \ \ t \in [0,T],
\end{equation}
where $\phi \in L^2(\tilde S, \P, \Gb)$ and $L^{U} \in \mathcal{M}^2_0(\P, \Gb)$ is strongly orthogonal to $\mathcal{I}^2(\tilde S, \P, \Gb)$.\\
The benchmarked cumulative cost process is
\begin{align*}
	C_t^{\bar \delta^T} =& n e^{- (\alpha + \gamma) T} \esp{p(Z_T) q(Y_T) g_T } + \int^t_0 (n - N_{u-}) e^{- \alpha T - \gamma (T-u)} q^{-1}(Y_u) \ud L_u^U\\
			&+ \int^t_0 U_{u-} e^{- \alpha T - \gamma (T-u)} q^{-1}(Y_u) \ud M_u, 
\end{align*}
for $t \in [0,T]$, where $M$ is given by (\ref{eq: process M}).

Now we consider the simplest case when the payoff is given by a polynomial function of the state variable, i.e.
\[
	g_T = g(Z_T) , \ \ \ \ \text{ with $g$ polynomial function}.
\]
In this case the pricing formula is reduced to
\begin{equation}\label{eq: pure endow polynomial}
	V^T_t = (n - N_t) e^{- (\gamma + \alpha) ( T - t)} \frac{\hat {pqg}_{(t,T)}(Z_t)}{p(Z_t) q(Y_t)}, \ \ \ \ t \in [0,T].
\end{equation}
We note that this includes the realistic cases for an insurance contract with constant payoff $g_T = k$, $k \in \R^+$, or the one with an index-linked payoff, e.g. proportional to the survival index at time $T$, that is $g_T = k I_T = k e^{-\gamma T} q(Y_T)$, $k \in \R^+$. 
In this case, we have
\begin{equation}\label{eq: discounted value process}
	U_t = \hat {pqg}_{(t,T)}(Z_t), \ \ \ \ t \in [0,T].
\end{equation}
Lemma \ref{lemma: decomposition hat p} applied to (\ref{eq: discounted value process}), (\ref{eq: OIS bond definition}) and (\ref{eq: def longevity bond}) leads to the following decompositions
\[
	\hat{pqg}_{(t,T)}(Z_t) =\hat{pqg}_{(0,T)}(Z_0) + \int^t_0 \nabla_x \hat{pqg}_{(u, T)} (Z_u)^\top \sigma(X_u) \ud W_u,
\]
\[
	(S^*_t)^{-1} P(t, T) = e^{- \alpha T} \hat{p}_{(0,T)} (X_0) + \int^t_0 e^{-\alpha T} \nabla_x \hat{p}_{(u, T)} (Z_u)^\top \sigma(X_u) \ud W_u,
\]
\[
	(S^*_t)^{-1} P^l(t, T) = e^{- (\alpha + \gamma) T} \hat{pq}_{(0,T)} (Z_0)  + \int^t_0 e^{-(\alpha + \gamma) T} \nabla_x \hat{pq}_{(u, T)} (Z_u)^\top \sigma(X_u) \ud W_u,
\]
where $t \in [0,T]$.
We define the 2-dimensional square matrix process $\theta:=(\theta_t)_{t \in [0,T]}$
\begin{equation}\label{eq: matrix A}
	\theta_t := \left[ \begin{array}{cc}
	 e^{-\alpha T} \sigma (X_t)^\top \nabla_x \hat{p}_{(t, T)} (Z_t), & e^{-(\alpha + \gamma) T} \sigma(X_t)^\top \nabla_x \hat{pq}_{(t, T)} (Z_t) \end{array} \right],
\end{equation}
and the 2-dimensional vector process $\phi := (\phi_t)_{t \in [0,T]}$
\[
	\phi_t = \left( \begin{array}{c}
		\phi_t^1 \\ \phi_t^2 \end{array} \right),
\]
satisfying
\[
	\theta_t \phi_t  = \sigma (X_t)^\top \nabla \hat{pqg}_{(t, T)} (Z_t),
\]
for all $t \in [0,T]$. Providing that the matrix $\theta_t$ is a.s. invertible for all $t \in [0,T]$, we have
\[
	\phi_t = \theta_t^{-1} \sigma (X_t)^\top \nabla \hat{pqg}_{(t, T)} (Z_t),
\]
for all $t \in [0,T]$.
Then 
\[
	\hat{pqg}_{(t,T)}(Z_t) =\hat{pqg}_{(0,T)}(Z_0) + \int_0^t \phi^1_u \ud \left((S^*_u)^{-1} P(u,T)\right) + \int_0^t \phi^2_u \ud \left((S^*_u)^{-1} P^l(u,T)\right),
\]
hence for $t \in [0,T]$, the benchmarked risk-minimizing strategy is given by
\[
	\bar \delta^T_t = \left((n - N_{t-}) e^{-\alpha T - \gamma (T - t)} q^{-1}(Y_t) \phi^1_t, (n - N_{t-}) e^{-\alpha T - \gamma (T - t)} q^{-1}(Y_t) \phi^2_t \right),
\]
and the benchmarked cumulative cost process is given by
\[
	C_t^{\bar \delta^T} = n e^{- (\alpha + \gamma) T} \hat{pqg}_{(0,T)}(Z_0) + \int^t_0 \hat{pqg}_{(u,T)}(Z_{u}) e^{- \alpha T - \gamma (T-u)} q^{-1}(Y_u) \ud M_u.
\]

If now we assume that the payoff is a generic continuous function of the state variable, i.e.
\[
	g_T  = g(Z_T) , \ \ \ \ \text{ with $g$ continuous function on } E,
\]
then it is not always possible to find an explicit form of the conditional expectation as in the polynomial case.
This class includes a large family of longevity linked products, e.g. options on survival index or longevity bond.
However, providing that the state space $E$ is compact, we can always find a uniform polynomial approximation $\{ g_m \}_{m \in \N}$ of $g$ on E, i.e. 
\begin{equation}\label{eq: approximation g}
	\norm{g_m - g}_\infty \xrightarrow[]{m \rightarrow \infty} 0 \text{ \ \ \ \  on \ \ } E,
\end{equation}
where the norm $\norm{ \cdot }_\infty$ is defined by
\[
	\norm{f}_\infty := \sup_{\substack{x \in E \\ \norm{x} = 1}} |f(x)|,
\]
for any $f \in \mathcal{C}(E)$. Proposition \ref{prop: approximation pure endowment} shows that the sequence of pricing formulas related to $\{ g_m \}_{m \in \N}$ provides a good approximation of the one related to $g$.

\begin{lemma}\label{lemma: approximation of conditional expectation}
	Let $\{ g_m \}_{m \in \N}$ be a uniform polynomial approximation of the continuous function $g$ on $E$ as in (\ref{eq: approximation g}). Then
	\[
		\sup_{t \in [0,T]} \left|\condespf{g_m(Z_T) - g(Z_T)} \right|\xrightarrow[]{m \rightarrow \infty} 0, \ \ \ \ \text{a.s.},
	\]
	\[
		\sup_{t \in [0,T]} \norm{ \condespf{g_m(Z_T) - g(Z_T)} }_{L^p(\Omega, \P)} \xrightarrow[]{m \rightarrow \infty} 0,
	\]
	for all $p \geqslant 1$.
\end{lemma}

\begin{proof}
	We first prove the a.s. approximation,
	\begin{align*}
		& \sup_{t \in [0,T]} \left| \condespf{g_m(Z_T) - g(Z_T)} \right| \leqslant \sup_{t \in [0,T]} \condespf{ \left| g_m(Z_T) - g(Z_T) \right| \ } \\
		 &\leqslant \sup_{t \in [0,T]} \condespf{ \norm{ g_m - g }_\infty \ } = \norm{g_m - g}_\infty \xrightarrow[]{m \rightarrow \infty} 0, \ \ \ \ \text{a.s.}.
	\end{align*} 
	Similarly we have the $L^p(\Omega, \P)$ approximation uniformly in $t \in [0, T]$ for any $p \geqslant 1$,
	\begin{align*}
		& \sup_{t \in [0,T]} \esp{\left| \condespf{g_m(Z_T) - g(Z_T)} \right|^p} \leqslant \esp{ \left( \sup_{t \in [0,T]}  \left| \condespf{g_m(Z_T) - g(Z_T)} \right| \right)^p} \\
		 &\leqslant \norm{g_m - g}^p_\infty \xrightarrow[]{m \rightarrow \infty} 0.
	\end{align*} 
\end{proof}
\ \\ \
We set $V^T :=(V^T_t)_{t \in [0,T]}$ where
	\begin{equation}\label{eq: VT}
		V^T_t = (n - N_t) e^{- (\gamma + \alpha) ( T - t)} \frac{\condespf{p(Z_T) q(Y_T) g(Z_T) }}{p(Z_t) q(Y_t)}, 
	\end{equation}
for $t \in [0,T]$.

\begin{prop}\label{prop: approximation pure endowment}
	Let $\{ g_m \}_{m \in \N}$ be a uniform polynomial approximation of the continuous function $g$ on $E$ as in (\ref{eq: approximation g}). For every $m \in \N$ we consider 
	\[
		V^{T,m} := \left( V^{T,m}_t \right)_{t \in [0,T]} = \left((n - N_t) e^{- (\gamma + \alpha) ( T - t)}  \frac{\hat {pqg_m}_{(t,T)}(Z_t)}{p(Z_t) q(Y_t)}\right)_{t \in [0,T]}.
	\]
	Then $\left\{V^{T,m}\right\}_{m \in \N}$  provides both a pathwise and $L^p(\Omega, \P)$ approximation of $V^T$ in (\ref{eq: VT}) for any $p \geqslant 1$ uniformly in $t \in [0,T]$, i.e.
	\[
		\sup_{t \in [0,T]} \left|V^{T,m}_t - V^T_t \right|\xrightarrow[]{m \rightarrow \infty} 0, \ \ \ \ \text{a.s.},
	\]
	\[
		\sup_{t \in [0,T]} \norm{ V^{T,m}_t - V^T_t }_{L^p(\Omega, \P)} \xrightarrow[]{m \rightarrow \infty} 0.
	\]
\end{prop}

\begin{proof}
	Straightforward from Lemma \ref{lemma: approximation of conditional expectation}.
\end{proof}

Now we prove that both the benchmarked risk-minimizing strategies and benchmarked cumulative cost processes associated to $\{ g_m \}_{m \in \N}$ provide a good approximation of the ones associated to $g$ as well.

\begin{lemma}\label{lemma: convergence phi and LU}
	Let $\{ g_m \}_{m \in \N}$ be a uniform polynomial approximation of the continuous function $g$ on $E$ as in (\ref{eq: approximation g}) and for every $m \in \N$ we consider
	\begin{equation}\label{eq: Um}
		U_m := \left((U_m)_t\right)_{t \in [0,T]} = \left(\hat{pqg_m}_{(t,T)} (Z_t)\right)_{t \in [0,T]},
	\end{equation}
	with the following Galtchouk-Kunita-Watanabe decomposition
	\[
		\left(U_m\right)_t = \left(U_m\right)_0 + \int ^t_0 (\phi_m)_u^\top \ud \tilde S_u + L^{U_m}_t, \ \ \ \ t \in [0,T].
	\]
	Let $\phi$ and $L^U$ be the two processes given by the Galtchouk-Kunita-Watanabe decomposition of $U$ in (\ref{eq: gkw decomposition of U}) with respect to $\tilde S$, then
	\begin{equation}\label{eq: convergence phi}
		\norm{\phi - \phi_m }_{L^2(\tilde S, \P, \Gb)}\xrightarrow[]{m \rightarrow \infty} 0,
	\end{equation}
	and 
	\begin{equation}\label{eq: convergence L}
		\norm{L^U - L^{U_m}}_{M^2_0(\P, \Gb)} \xrightarrow[]{m \rightarrow \infty} 0.
	\end{equation}
	If the matrix process $\theta$ defined in (\ref{eq: matrix A}) is such that, for all $t \in [0,T]$, $\theta_t$ is a.s. invertible with $\theta^{-1} \in L^2 \left(\Omega \times [0,T], \P \otimes \ud t \right)$, then
	\begin{equation}\label{eq: convergence phi with A invertible}
		\norm{\phi - \phi_m}_{L^1 \left(\Omega \times [0,T], \P \otimes \ud t \right)} \xrightarrow[]{m \rightarrow \infty} 0,
	\end{equation}
	and $L^U = L^{U_m} = 0$ for all $m \in \N$.
\end{lemma}

\begin{proof}
	Proposition \ref{prop: approximation pure endowment} gives in particular the following convergence in $M^2_0(\P, \Gb)$,
	\[
		\norm{ U - U_m }^2_{M^2_0(\P, \Gb)} = \sup_{t \in [0,T]}	\esp{\left( U_t - (U_m)_t \right)^2}\xrightarrow[]{m \rightarrow \infty} 0.
	\]
	Since $L^U$ and $\left( L^{U_m} \right)_{m \in \N}$ are strongly orthogonal to the space $\mathcal{I}(\tilde S, \P, \Gb)$, we have
	\begin{align*}
		&\norm{ \int ^\cdot_0 \left( \phi_u^\top - (\phi_m)_u^\top \right) \ud \tilde S_u + \left(L^U - L^{U_m}\right) }^2_{M^2_0(\P, \Gb)}\\
		=& \norm{ \int ^\cdot_0 \left( \phi_u^\top - (\phi_m)_u^\top \right) \ud \tilde S_u}^2_{M^2_0(\P, \Gb)} + \norm{ L^U - L^{U_m} }^2_{M^2_0(\P, \Gb)}\\
		=&\norm{\phi - \phi_m}^2_{L^2(\tilde S, \P, \Gb)} + \norm{ L^U - L^{U_m} }^2_{M^2_0(\P, \Gb)},
	\end{align*}
	which implies that
	\begin{equation}\label{eq: converg L}
		\norm{L^U - L^{U_m} }_{M^2_0(\P, \Gb)} \xrightarrow[]{m \rightarrow \infty} 0,
	\end{equation}
	and
	\[
		 \norm{\phi - \phi_m }_{L^2(\tilde S, \P, \Gb)} \xrightarrow[]{m \rightarrow \infty} 0.
	\]
	Moreover, for every $m \in \N$, the It\^{o} isometry yields
	\begin{align*}
		&\norm{\phi - \phi_m }_{L^2(\tilde S, \P, \Gb)}\\
		=&\esp{\int^T_0  \left( \phi_u- (\phi_m)_u \right)^\top \ud [\tilde S]_u  \left( \phi_u - (\phi_m)_u \right)}\\
		=& \esp{\int^T_0  \left( \phi_u- (\phi_m)_u \right)^\top \theta_u^\top  \ud[W]_u  \theta_u  \left( \phi_u - (\phi_m)_u \right)} \\
		=&\esp{    \int^T_0 \left( \phi_u- (\phi_m)_u \right)^\top \theta_u^\top  \theta_u  \left( \phi_u - (\phi_m)_u \right)\ud u} \\
		=& \norm{ \theta (\phi - \phi_m) }_{L^2 \left(\Omega \times [0,T], \P \otimes \ud t \right)}.
	\end{align*}
	If the matrix $\theta_u$ is invertible for all $u \in [0,T]$ a.s. with 
	\[
	\theta^{-1} \in {L^2 \left(\Omega \times [0,T], \P \otimes \ud t \right)},
	\]
	then Cauchy–Schwarz inequality leads to 
	\begin{align*}
			&\esp{\int^T_0 \left| (\phi)_u - (\phi_m)_u \right| du}\\ 
		\leqslant& \norm{ \theta (\phi - \phi_m) }_{L^2 \left(\Omega \times [0,T], \P \otimes \ud t \right)} \cdot \norm{\theta^{-1} }_{L^2 \left(\Omega \times [0,T], \P \otimes \ud t \right)}\\
		&\xrightarrow[]{m \rightarrow \infty} 0.
	\end{align*}
	In particular we have $L^{U_m} = 0$ for all $m \in \N$ by Lemma \ref{lemma: decomposition hat p}. Then (\ref{eq: converg L}) yields $L = 0$.
\end{proof}

\begin{rem}\label{rem: martingale representation}
	We note that if $g$ is given by a continuous function, via a convergence argument (similar to the one in Lemma \ref{lemma: convergence phi and LU}) we obtain by Lemma \ref{lemma: decomposition hat p} that the Galtchouk-Kunita-Watanabe decomposition of $U$ with projection on the subspace $\mathcal{I}(W, \P, \Gb)$ is given by
	\[
		U_t = U_0 + \int^t_0 \psi^\top_u \ud W_u, \ \ \ \ t \in [0,T],
	\]
	where $\psi := (\psi_t)_{t \in [0,T]}$ is a predictable $W$-integrable vector process, i.e. $U$ contains no orthogonal term even without the assumption $\Fb = \Fb^W$.
\end{rem}

\begin{prop}
	Let $\{ g_m \}_{m \in \N}$ be a uniform polynomial approximation of the continuous function $g$ on $E$ as in (\ref{eq: approximation g}). If $\bar \delta^T$ and $C^{\bar \delta^T}$ are respectively the benchmarked risk-minimizing strategy and benchmarked cumulative cost process associated to $g$, $\bar \delta^T_m$ and $C^{\bar \delta^T_m}$ are the ones associated to $g_m$, then
	\begin{equation}\label{eq: convergence strategy}
		\norm{ \bar \delta^T - \bar \delta^T_m }_{L^2(\tilde S, \P, \Gb)} \xrightarrow[]{m \rightarrow \infty} 0,
	\end{equation}
	\begin{equation}\label{eq: convergence cost}
		\norm{ C^{\bar \delta^T} - C^{\bar \delta^T_m} }_{M^2_0(\P, \Gb)} \xrightarrow[]{m \rightarrow \infty} 0.
	\end{equation}
	If furthermore the matrix process $\theta$ given by (\ref{eq: matrix A}) is a.s. invertible, then
	\begin{equation}\label{eq: convergence strategy with A invertible}
		\norm{\bar \delta^T - \bar \delta^T_m}_{L^1 \left(\Omega \times [0,T], \P \otimes \ud t \right)} \xrightarrow[]{m \rightarrow \infty} 0.
	\end{equation}
\end{prop}

\begin{proof}
	(\ref{eq: convergence strategy}) and (\ref{eq: convergence strategy with A invertible}) are immediate consequence of (\ref{eq: convergence phi}) and (\ref{eq: convergence phi with A invertible}) in Lemma \ref{lemma: convergence phi and LU}. Now we prove the convergence in $M^2_0(\P, \Gb)$ of the benchmarked cumulative cost process. We note
	\begin{align*}
				&\esp{\left( C_t^{\bar \delta^T} - C_t^{\bar \delta^T_m} \right)^2}\\
		\leqslant& \ c_m + 2 \esp{\int^t_0 \left( (n - N_{u-}) e^{- \alpha T - \gamma (T-u)} q^{-1}(Y_u) \right)^2 \ud [L^U - L^{U_m}]_u}\\
				&+ 2 \esp{\int^t_0 \left( (U_{u} - (U_m)_{u}) e^{- \alpha T - \gamma (T-u)} q^{-1}(Y_u) \right)^2 \ud [M]_u}, 
	\end{align*}
	for every $t \in [0, T]$, where 
	\[
		c_m = 2 \left( C^{\bar \delta^T}_0 - C^{\bar \delta^T_m}_0 \right) = 2 \left( n e^{- (\alpha + \gamma) T} \esp{p(Z_T) q(Y_T) g(Z_T) - p(Z_T) q(Y_T) g_m(Z_T)} \right).
	\]
	Clearly 
	\[
		c_m \xrightarrow[]{m \rightarrow \infty} 0.
	\]
	For the first addend, thanks to the compactness of the state space we have 
	\begin{align*}
		2 \esp{\int^t_0 \left( (n - N_{u-}) e^{- \alpha T - \gamma (T-u)} q^{-1}(Y_u) \right)^2 \ud [L^U - L^{U_m}]_u} \leqslant \bar c \esp{(L^U_t - L^{U_m}_t)^2},
	\end{align*}
	for every $t \in [0, T]$, where $\bar c$ is a suitable constant. This quantity turns to zero uniformly in $t$ thanks to (\ref{eq: convergence L}) in Lemma \ref{lemma: convergence phi and LU}.\\
	For the second addend, since $(U^m)_{m \in \N}$ provides a pathwise approximation of $U$ uniformly in $t \in [0,T]$, the dominated convergence theorem together with the boundedness of the integrand process yields
	\[
		2 \esp{\int^t_0 \left( (U_{u} - (U_m)_{u}) e^{- \alpha T - \gamma (T-u)} q^{-1}(Y_u) \right)^2 \ud [M]_u}\xrightarrow[]{m \rightarrow \infty} 0, \ \ \ \ t \in [0,T],
	\]
	that concludes the proof.
\end{proof}

\subsection{Term insurance}\label{sec: term insurance}

A term insurance contract gives a positive payoff in the case of a policyholder's decease before the term $T$ of contract. The payment process $R := (R_t)_{t \in [0, T]}$ is assumed to be $\Fb$-predictable and square integrable. The amount paid at $T$ to the $i$-th policyholder is given by
\[
	\Ind{0 < \tau^i \leqslant T} R_{\tau^i},
\]
for $i = 1,...,n$. In the case of a homogeneous portfolio of policies, we have
\[
	\sum_{i=1}^n \Ind{0 < \tau^i \leqslant T} R_{\tau^i}.
\]
The associated benchmarked payment process $A$ is
\[
	A_t = \sum_{i=1}^n \int_0^t (S^*_u)^{-1} \ud D_u = \sum_{i=1}^n (S^*_{\tau^i})^{-1} \Ind{0 < \tau^i \leqslant t} R_{\tau^i},
\]
for $t \in [0,T]$. 

We denote by $V^\tau := (V^\tau_t)_{t \in [0,T]}$ the price process associated to a homogeneous portfolio of term insurance contracts. The real world pricing formula (\ref{eq: pricing formula dividend}) together with (\ref{eq: benchmark process}) and (\ref{eq: mortality intensity}) yields
\begin{align*}
	V^\tau_t = &S^*_t \condespg{\sum_{i=1}^n (S^*_{\tau^i})^{-1} \Ind{t < \tau^i \leqslant T} R_{\tau^i}}\\
	= & \sum_{i=1}^n S^*_t \condespg{ (S^*_{\tau^i})^{-1} \Ind{t < \tau^i \leqslant T} R_{\tau^i}}\\
	= & \sum_{i=1}^n \Ind{\tau^i > t } S^*_t \condespf{ \int^T_t (S^*_u)^{-1} R_u e^{- \int^T_t \mu_u \ud u} \mu_u \ud u}\\
	= & (n - N_t) e^{(\gamma + \alpha)t} \frac{\condespf{\int^T_t  e^{-(\gamma + \alpha)u} R_u p(Z_u) q(Y_u) \mu_u \ud u}}{p(Z_t) q(Y_t)} \\
	= & (n - N_t) e^{(\gamma + \alpha)t} \frac{ \condespf{ \int^T_t e^{-(\gamma + \alpha)u} R_u p(Z_u) \left(\gamma q(Y_u) - \nabla q(Y_u)^\top \bar b(Z_u)\right)\ud u }}{p(Z_t) q(Y_t)},
\end{align*}
where in the third equality we use Proposition 5.5 of \cite{Bar} combined with Corollary 5.1.3 of \cite{Bie-Rut}. We note that Corollary 5.1.3 of \cite{Bie-Rut} requires that $R$ is a bounded process, but this hypothesis can be easily relaxed by using a localization argument together with the dominated convergence theorem for conditional expectation if $R$ is sufficiently integrable.

Now we assume
\[
	R_t = R(Z_t),
\]
for $t \in [0,T]$, with $R$ a continuous function on the compact state space $E$. Then the stochastic Fubini-Tonelli Theorem yields
\begin{align*}
	V^\tau_t = &(n - N_t) e^{(\gamma + \alpha)t} \frac{\condespf{\int^T_t  e^{-(\gamma + \alpha)u} R(Z_u) p(Z_u)  (\gamma q(Y_u) - \nabla q(Y_u)^\top \bar b(Z_u)) \ud u}}{p(Z_t) q(Y_t)}\\
		  = & (n - N_t) e^{(\gamma + \alpha)t} \frac{\int^T_t e^{-(\gamma + \alpha)u} \condespf{ R(Z_u) p(Z_u) (\gamma q(Y_u) - \nabla q(Y_u)^\top \bar b(Z_u)) }\ud u}{p(Z_t) q(Y_t)},
\end{align*}
for $t \in [0,T]$. As before, this expression can be approximated by explicated pricing formulas related to polynomial payoff.

\begin{prop}\label{prop: approximation term insurance}
	Let $\left\{ R_m \right\}_{m \in \N}$ be a sequence of polynomials functions which approximates uniformly the continuous function $R$ on $E$. For every $m \in \N$, we consider $V^{\tau,m} := (V^{\tau,m}_t)_{t \in [0,T]}$ with
	\begin{align*}
		V^{\tau, m}_t := & (n - N_t) e^{(\gamma + \alpha)t} \frac{\int^T_t e^{-(\gamma + \alpha)u}  \condespf{ \gamma r_m (Z_u) - s_m(Z_u)}\ud u}{p(Z_t) q(Y_t)}\\
				 = & (n - N_t) e^{(\gamma + \alpha)t} \frac{\int^T_t e^{-(\gamma + \alpha)u} \left( \gamma {\hat r}_{m(t,u)} (Z_t) - \hat{s}_{m(t,u)} (Z_t) \right) \ud u}{p(Z_t) q(Y_t)},
	\end{align*}
	for every $t \in [0,T]$, where the polynomial functions $r_m$ and $s_m$ are given respectively by $r_m := R_m p q$ and $s_m:= R_m p \left( \nabla q^\top \bar b \right)$. Then $\left\{ V^{\tau, m} \right\}_{m \in \N}$ provide both a pathwise and $L^p(\Omega, \P)$ approximation of $V^\tau$ uniformly in $t \in [0,T]$.
\end{prop}

\begin{proof}
	Analogous to Proposition \ref{prop: approximation pure endowment}.
\end{proof}

\subsection{Annuity}

An annuity is a continuous cash stream paid by the insurer as long as the policyholder is alive. We denote its cumulated payoff value up to time $t$ by $C_t$. The process $C := (C_t)_{t \in [0,T]}$ is assumed to be a right continuous increasing $\Fb$-adapted and square integrable process, with $C_0 = 0$ and $C_{T-} = C_T$. 
The total payoff at $T$ associated to the $i$-th policyholder is given by
\begin{equation}\label{eq: annuity payoff}
	\int_{]0,T]} (1 - H^i_u) d C_u = \int_{]0,T]} \I_{ \{\tau^i > u\} } d C_u = C_T \I_{ \{\tau^i > T\} } + C_{\tau^i-} \I_{ \{0 < \tau^i \leqslant T\} },
\end{equation}
the total payoff at $T$ of a homogeneous portfolio of annuity contracts is
\[
	\sum_{i=1}^n \int_{]0,T]} (1 - H^i_u) d C_u.
\]
The benchmarked cumulated payment process at time $t$ with $t \in [0,T]$ is
\[
	A_t = \sum_{i=1}^n \int^t_0 (S^*_u)^{-1} (1 - H^i_u) d C_u.
\]

Let $V^C := (V^C_t)_{t \in [0,T]}$ denote the price process given by the real world pricing formula (\ref{eq: pricing formula dividend}) for a homogeneous portfolio of annuity contracts. By (\ref{eq: benchmark process}) and (\ref{eq: mortality intensity}) we have at $t \in [0,T]$
\begin{align*}
	V^C_t :=  & S^*_t \condespg{\sum_{i=1}^n \int^T_t (S^*_u)^{-1} (1 - H_u) d C_u}\\
	=& \sum_{i=1}^n S^*_t \condespg{ \int^T_t (S^*_u)^{-1} (1 - H_u) d C_u}\\
	=& \sum_{i=1}^n \Ind{\tau^i > t}S^*_t \condespf{\int_{ ] t , T]} (S^*_u)^{-1} e^{- \int^u_t \mu_u \ud u} d C_u}\\
	=& (n - N_t) e^{(\gamma + \alpha)t} \frac{\condespf{\int^T_t e^{-(\gamma + \alpha)u} p(Z_u) q(Y_u) d C_u}}{p(Z_t) q(Y_t)},
\end{align*}
where in the third equality we use Proposition 5.5 of \cite{Bar} and Proposition 5.1.2 of \cite{Bie-Rut}. Proposition 5.1.2 of \cite{Bie-Rut} requires that the process $C$ is bounded. As in Section \ref{sec: term insurance}, this condition can be relaxed using a localization argument combined with the theorem of dominated convergence for conditional expectation.

\begin{rem}\label{rem: annuity sum of pure endowment and term insurance}
	We stress that under our assumption, if $C$ is furthermore a continuous process, then it is also an $\Fb$-predictable process. Therefore, according to (\ref{eq: annuity payoff}) we have 
	\[
		\sum_{i=1}^n \int_{]0,T]} (1 - H^i_u) d C_u = C_T (n - N_T) + \sum_{i=1}^n C_{\tau^i} \I_{ \{0 < \tau^i \leqslant T\} }.
	\]
	That is, a homogeneous annuity portfolio can be considered as the sum of a homogeneous pure endowment portfolio and a homogeneous term insurance portfolio as defined in Section \ref{sec: pure endowment} and \ref{sec: term insurance} respectively, where $g_T = C_T$ and $R = C$. In particular, the linearity of the pricing formula yields
	\[
		V^C = V^T + V^\tau.
	\]
\end{rem}

If now we assume 
\[
	C_t = \bar C(Z_t),
\]
for $t \in [0,T]$, with $\bar C$ a continuous function on the compact state space $E$, then we have the following proposition.

\begin{prop}
	Let $\{ C_m \}_{m \in \N}$ be a sequence of polynomials which approximates uniformly the continuous function $\bar C$ on $E$. For every $m \in \N$, we consider $V^{C,m} := (V^{C,m}_t)_{t \in [0,T]}$ with
	\begin{align*}
		V^{C,m} :=& V^{T,m} + V^{\tau,m},
	\end{align*}
	where $V^{T,m}$ and $V^{\tau,m}$ are respectively defined in Proposition \ref{prop: approximation pure endowment} and \ref{prop: approximation term insurance} with 
	\[
		g_m = R_m = C_m,
	\] 
	for all $m \in \N$.
	Then $\{ V^{C,m} \}_{m \in \N}$ provides both a pathwise and $L^p(\Omega, \P)$ approximation of $V^C$ uniformly in $t \in [0,T]$.
\end{prop}

\begin{proof}
	Straightforward from Remark \ref{rem: annuity sum of pure endowment and term insurance}, Proposition \ref{prop: approximation pure endowment} and Proposition \ref{prop: approximation term insurance}.
\end{proof}

\section{A numerical example}\label{sec: simulation}

We now consider a numerical example with calibration to real data. Set $m_1 = m_2 = 1$, throughout this section we assume that $E^X =[-1,1]$ and $Z_0 = (X_0, Y_0)^\top = 0$. In particular $E^Y$ is also bounded, see Proposition \ref{prop: boundness Y}.

For a more detailed study of polynomial diffusions on unit ball we refer to \cite{Lar-Pul}. In view of Theorem 2.1 of \cite{Lar-Pul} we consider the following model dynamics,
\[
	\ud Z_t = 
	\ud \left( \begin{array}{c}
				X_t \\
				Y_t
			\end{array} \right) =
		\left( \begin{array}{cc}
				\Psi & 0 \\
				d & \kappa 
				\end{array} \right)
		\left(
			\left( \begin{array}{c}
						b \\
						\eta
					\end{array} \right)
			- 
			\left( \begin{array}{c}
						X_t \\
						Y_t
					\end{array} \right)
		\right) \ud t
		+ \sigma(X_t) \ud W_t,
\]
where $W$ is a 1-dimensional Brownian motion,
\[
	a(x)= \sigma(x)^2 = \sigma (1 - x^2), \ \ \ \ \text{ for all } x \in \mathcal [-1, 1],
\]
with $\sigma > 0$, $\Psi, b, d , \kappa, \eta \in \R$ 
and the parameters satisfy the following condition 
\[
	b \Psi x - \Psi x^2  \leqslant 0, \ \ \ \ \text{ for } x \in \{1, -1\},
\]
equivalent to $|b \Psi| \leqslant \Psi$ or
\begin{equation}
	\left\{\begin{array}{c}
		|b| \leqslant 1, \\
		\Psi \geqslant 0.
	\end{array} \right.
\end{equation}
In particular the dynamics of component $X$ is given by
\[
	\ud X_t = \Psi (b - X_t) \ud t + \sigma \sqrt{1 - X^2_t} \ud W_t.
\]

Furthermore we assume that the polynomials $p$ and $q$ are both linear and positive on $E$, i.e.
\begin{equation}\label{eq: p}
	p(x) = \rho + c x, \ \ \ \ p > 0 \text{ on } E^X,
\end{equation}
\begin{equation}\label{eq: q}
	q(y) = \delta + \nu y, \ \ \ \ q > 0 \text{ on } E^Y,
\end{equation}
where $\rho, \delta, \nu, c \in \R$. A similar specification for $p$ can be found in \cite{Fil-Lar2}.

Under these assumptions we have
\[
	\frac{1}{S^*_t} = e^{- \alpha t} \left( \rho + c X_t \right),
\]
\[
	P(t,T) = \frac{\left( \rho + c b \right) e^{- \alpha (T - t)} + c e^{-(\alpha + \Psi)(T - t)} (X_t - b)}{\rho + cp X_t},
\]
\[
	r_t = \alpha - \frac{c \Psi (b - X_t)}{\rho + c X_t},
\]
\[
	I_t = e^{- \gamma t} \left( \delta + \nu Y_t \right),
\]
\[
	\mu_t = \gamma - \frac{\nu \left ( d b + \kappa \eta - d X_t - \kappa Y_t \right) }{\delta + \nu Y_t},
\]
\[
	\tilde P(t, T) = e^{- \gamma T - \alpha (T - t)} \frac{\hat{pq}_{(t, T)}(Z_t)}{\rho + c X_t},
\]
with dynamics
\[
	\frac{\ud (S^*_t)^{-1}}{(S^*_t)^{-1}} = - r_t \ud t - \lambda_t \ud W_t,
\]
where $\lambda_t = - \sqrt{a(X_t)} c / (\rho + c X_t)$,
\[
	\frac{\ud P(t, T)}{P(t, T)} = (r_t + \nu(t,T) \lambda_t) \ud t + \nu(t, T) \ud W_t,
\]
where $\nu(t, T) = \sqrt{a(X_t)} \nabla P(t, T) / P(t, T)$,
\[
	\frac{\ud I_t}{I_t} = - \mu_t \ud t.
\]

For the sake of simplicity, we calibrate our model to the inverse of benchmark portfolio and the longevity index. The benchmark portfolio can be identified with a sufficiently diversified portfolio such as Morgan Stanley capital weighted world stock accumulation index, called MSCI world index (see discussion in \cite{Pla2} and \cite{Pla3}). For the second one we take data from LLMA index related to German population. The sample period ranges from January 1970 to January 2013 with 517 monthly observations of MSCI world index and 44 annual observations of Germany male graduated initial rate of mortality published by LLMA relating to the cohort of male population aged 20 in 1970. 

The following table reports the summary statistics of the two data sets. The inverse of benchmark portfolio data is reported in basis points and longevity index data in percentages.

\begin{center}
	\begin{tabular}{c c c c c c}
		\hline
		\ 	 	  & Mean & Median & Std. & min & MAX \\
		\hline
		$1/\text{MSCI index}$ ($1/S^*$) & 38.612 & 19.468 & 35.397 & 5.9441 & 134.31 \\
		Longevity index ($I$) & 98.628 & 98.984 & 1.1337 & 95.606 & 99.803 \\
		\hline
	\end{tabular}
\end{center}

We denote the model parameter vector by $\Phi$ and the times of observation by $t_1, t_2, ..., t_N$, where $N= 517$. For every $t_k$ with $1 \leqslant k \leqslant N$ we have a 1- or 2-dimensional observation vector $v_{t_k}$. When both MSCI index and LLMA index are observable the measurement equation is given below. When only the MSCI index is observable, $v_{t_k}$ is reduced to the only first component. We have
\begin{align*}
	v_{t_k} &= f(Z_{t_k}, \Phi) + \varepsilon_{t_k} = \left[ \frac{1}{S^*_{t_k}} ,\ I_{t_k} \right]^\top + \varepsilon_{t_k}\\
		&= \left[ e^{- \alpha {t_k}} \left( \rho + c X_{t_k} \right), e^{- \gamma {t_k}} \left( \delta + \nu Y_{t_k} \right) \right]^\top + \varepsilon_{t_k}\\
		&= \Theta_{k0} + \Theta_{k1} Z_{t_k} + \varepsilon_{t_k},
\end{align*}
where
\[
	\Theta_{k0} = \left( \begin{array}{c}
							e^{- \alpha t_k} \rho \\
							e^{- \gamma t_k} \delta 
					\end{array} \right), \ \ \ \ \ \ 
	\Theta_{k1} = 	\left( \begin{array}{cc}
							e^{- \alpha t_k} c & 0 \\
							0 & e^{- \gamma t_k} \nu 
					\end{array} \right),
\]
and the measurement error vector is assumed\footnote{The same assumption can be found in \cite{Lun-All} and in \cite{Fil-Lar2}. In a more general case $\varepsilon_{t_k}$ is a random error vector with $\mathbb{E}[\varepsilon_{t_k}] = 0$.} to be
\[
	\varepsilon_{t_k} \sim N \left(0, \left( \begin{array}{cc}
									\sigma^2_1 & 0 \\
									0 & \sigma^2_2
								\end{array} \right) \right), 
\] 
where $\sigma^2_1$ indicate the measurement error variance associated to the inverse of benchmark portfolio and $\sigma^2_2$ the one associated to the longevity index.

In view of \cite{Lar-Pul} and under the assumption that the longevity index does not have relevant influence on the benchmark portfolio, the transition equation in discrete time of the first component $X$ of (unobserved) state variable for $t_k$ with $1 < k \leqslant N$ can be approximated\footnote{See e.g. \cite{Dua-Sim}.} by,
\[
	X_{t_{k}} = \mathbb{E}\left. \left[ X_{t_{k}} \right| X_{t_{k-1}} \right] + \text{Var} \left. \left[ X_{t_{k}} \right| X_{t_{k-1}} \right]^{\frac{1}{2}} \eta_{t_k}, 
\]
where $\eta_{t_k}$ is a 1-dimensional error term of zero mean and unit variance, independent from $X_{t_{k-1}}$. 
We note that $\mathbb{E}\left. \left[ X_{t_{k}} \right| X_{t_{k-1}} \right]$ and $\text{Var} \left. \left[ X_{t_{k}} \right| X_{t_{k-1}} \right]^{\frac{1}{2}} \eta_{t_k}$ are (both conditionally and unconditionally) independent. Such approximation gives exact conditional and unconditional expectation and variance matrix of $X_{t_k}$. Furthermore, it follows from the property given in Proposition \ref{prop pol cond exp} that the conditional expectation $\mathbb{E}\left. \left[ X_{t_{k}} \right| X_{t_{k-1}} \right]$ is an affine function of $X_{t_{k-1}}$ and the conditional variance $\text{Var} \left. \left[X_{t_k} \right| X_{t_{k-1}} \right]$ is a second degree polynomial function of $X_{t_{k-1}}$. To be more precise,
\[
	\mathbb{E}\left. \left[ X_{t_{k}} \right| X_{t_{k-1}} \right] = \Phi_{k0} + \Phi_{k1}X_{t_{k-1}},
\]
where
\[
	\Phi_{k0} = b \left( 1 - e^ {- \Psi \left(t_k - t_{k-1}\right)} \right),
\]
\[
	\Phi_{k1} = e^ {-\Psi \left(t_k - t_{k-1}\right)}.
\]
Following \cite{Lun-All} we make the further assumption of approximating $\text{Var} \left. \left[ X_{t_{k}} \right| X_{t_{k-1}} \right]^{\frac{1}{2}} \eta_{t_k}$ with a normal distribution error term $u_{k-1}$ independent from $\mathbb{E}\left. \left[ X_{t_{k}} \right| X_{t_{k-1}} \right]$, 
\begin{equation}\label{eq: transition function X}
	X_{t_{k}} = \Phi_{k0} + \Phi_{k1}X_{t_{k-1}} + u_{k-1}.
\end{equation}
\[
	u_{k-1} \sim N(0, Q_{k-1}), \ \ \ \ \ 
	Q_{k-1} = \esp{\text{Var} \left. \left[X_{t_k} \right| X_{t_{k-1}} \right]}.
\]
The second component $Y$ of state variable for any time $t$ has an explicit solution depending on $X$,
\[
	Y_t = e^{-\kappa t} \int^t_0 \left( -d X_s + d b + \kappa \eta \right)e^{\kappa s} \ud s.
\]
For $t$ equal to an observation time of the longevity index, it can be approximated by
\[
	Y_t = e^{-\kappa t} \sum_{0 \leqslant t_{k_i} < t} \left( -d X_{t_{k_i}} + d b + \kappa \eta \right)e^{\kappa t_{k_i}} (t_{k_{i + 1}} - t_{k_i}),
\]
where $X_{t_{k_i}}$ are the monthly values calculated by the transition equation (\ref{eq: transition function X}) of $X$.


Since both the inverse of the benchmark portfolio $1/S^*$ and the longevity index $I$ are affine functions of the state variable, if we ignore for now the state space restrictions
\footnote{See \cite{Lun-All} for a more detailed discussion regarding this assumption.}, then we are in the case of a linear Gaussian state space model. Linear Kalman filter and maximum likelihood estimation are thus applicable under these approximations. We refer to \cite{Har} for a detailed description of the method. For the sake of simplicity, we apply linear Kalman filter and maximum likelihood estimation only to estimate parameters of the $X$ component. Since the longevity index can be considered as a linear regression of $X$ under our approximation, least squares estimation will be applied to estimate the remaining parameters of the $Y$ component.

Let $V_{t_k}$ denote the information available at time $t_k$ regarding the benchmark portfolio, i.e.
\[
	V_{t_k} = (v^1_{t_1}, v^1_{t_2},..., v^1_{t_k}).
\]
For $1 < k \leqslant N$ we denote
\[
	\hat X_{t_k| t_{k-1}} := \mathbb{E} \left. \left[ X_{t_k} \right| V_{t_{k-1}} \right], \ \ \ \ \Sigma_{t_k| t_{k-1}} := \text{Var}\left. \left[ X_{t_k} \right| V_{t_{k-1}} \right],
\]
where $\hat X_{t_k| t_{k-1}}$ is the optimal predictor of $X_{t_k}$ and $\Sigma_{t_k| t_{k-1}}$ is its mean square error. Analogously for $1 \leqslant k \leqslant N$ we denote
\[
	\hat X_{t_k} := \mathbb{E} \left. \left[ X_{t_k} \right| V_{t_k} \right], \ \ \ \ \Sigma_{t_k} := \text{Var}\left. \left[ X_{t_k} \right| V_{t_k} \right].
\]
For $1 < k \leqslant N$, the \emph{prediction step} of linear Kalman filter is given by
\[
	\hat X_{t_k| t_{k-1}} = \Phi_{k0} + \Phi_{k1} \hat X_{t_{k-1}},
\]
with mean square error
\[
	\Sigma_{t_k| t_{k-1}} = \Phi^2_{k1} \Sigma_{t_{k-1}} + Q_{k-1},
\]
and the \emph{update step} is given by 
\[
	\hat X_{t_k} = \hat X_{t_k| t_{k-1}} + \Sigma_{t_k| t_{k-1}} \Theta^1_{k1} \left(F_{t_k} \right)^{-1} w_{t_k},
\]
\[
	\Sigma_{t_k} = \left( \left(\Sigma_{t_k| t_{k-1}} \right)^{-1} + \left( \Theta^1_{k1} \right)^2 \sigma^{-2}_1 \right)^{-1},
\]
where
\[
	w_{t_k} = v^1_{t_k} - \mathbb{E}\left. \left[v^1_{t_k} \right| V_{t_{k-1}} \right] = v^1_{t_k} 
				- \left( \Theta^1_{k0} + \Theta^1_{k1} \hat X_{t_k| t_{k-1}} \right),
\]
\[
	F_{t_k} = \text{Var} (w_{t_k}) = \left( \Theta^1_{k1} \right)^2 \Sigma_{t_k | t_{k-1}} + \sigma^{2}_1.
\]
The (approximated) log-likelihood function is given by
\[
	\log L(v_{t_1}, v_{t_2},..., v_{t_N}; \Phi) = \sum_{k=1}^N -\log(2 \pi) - \frac{1}{2} \log |F_{t_k}| - \frac{1}{2} w_{t_k}^\top F_{t_k}^{-1} w_{t_k}.
\]
For $k = 1 + 12*h$ with $h = 0, ... , 43$\footnote{We note that longevity index is only annually observable, thus $Y$ can be updated only annually.} the approximated value of $Y_{t_k}$ is
\begin{equation}\label{eq: approx Y}
	\hat Y_{t_k} = e^{-\kappa t_k} \sum^{k-1}_{s=0} \left( -d \hat X_{t_s} + d b + \kappa \eta \right)e^{\kappa t_s} (t_{s + 1} - t_s).
\end{equation}

We set $\rho= 0.01, c = 0.006, \delta = 0.998, \nu = -0.00044$ such that conditions (\ref{eq: p}) and (\ref{eq: q}) are satisfied. In particular the assumption of $Z_0 = (X_0,Y_0)^\top = 0$ forces the value of $\rho$ and $\delta$ to be (almost) equal to the first value of the inverse of benchmark portfolio and longevity index respectively.\footnote{While our condition $Z_0 = (X_0,Y_0)^\top = 0$ limits the choice of $\rho$ and $\delta$, the values of $c$ and $\nu$ can be arbitrarily chosen, as soon as (\ref{eq: p}) and (\ref{eq: q}) are fulfilled. A different choice of $c$ and $\nu$ will result in a scaling of the state variable $Z$. See also Theorem 5 of \cite{Fil-Lar2}.}

The following table reports the calibrated parameters.
\begin{center}
	\begin{tabular}{c c c c c c}
		\hline
		$\Psi$ & $b$ & $\sigma$ & $d$ & $\kappa$ & $\eta$ \\
		\hline
		14.98581 & -0.79506 & 1.25299 & 5.18417 & -5.87517 & -5.05117 \\
		\hline
	\end{tabular}
\end{center}
The correspondent values of $\alpha, \gamma$, which adjust the short rate and mortality intensity level, and the log-likelihood value are reported below.
\begin{center}
	\begin{tabular}{c c c}
		\hline
		$\alpha$ & $\gamma$ & L \\
		\hline
		4.6068 & 0.0045607 & 2347.5 \\
		\hline
	\end{tabular}
\end{center}

\begin{figure}[!htbp]
	\centering
	\subfigure[Inverse of benchmark portfolio data and fit\label{subfig: Benchmark portfolio data and fit}]
	{\includegraphics[scale=0.31]{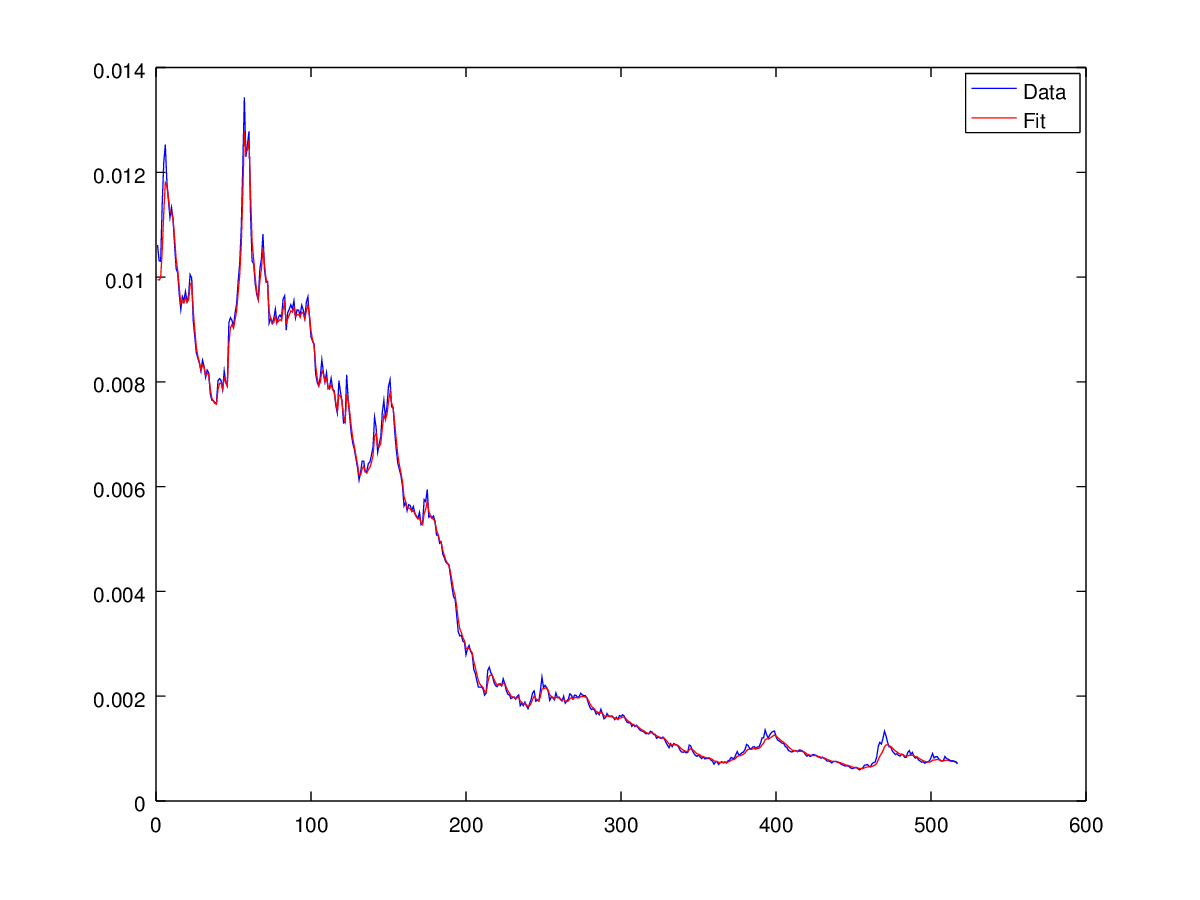}}
	\\ \ \\
	\subfigure[Pricing error generated by Kalman filter\label{subfig: Benchmark portfolio difference}]
	{\includegraphics[scale=0.31]{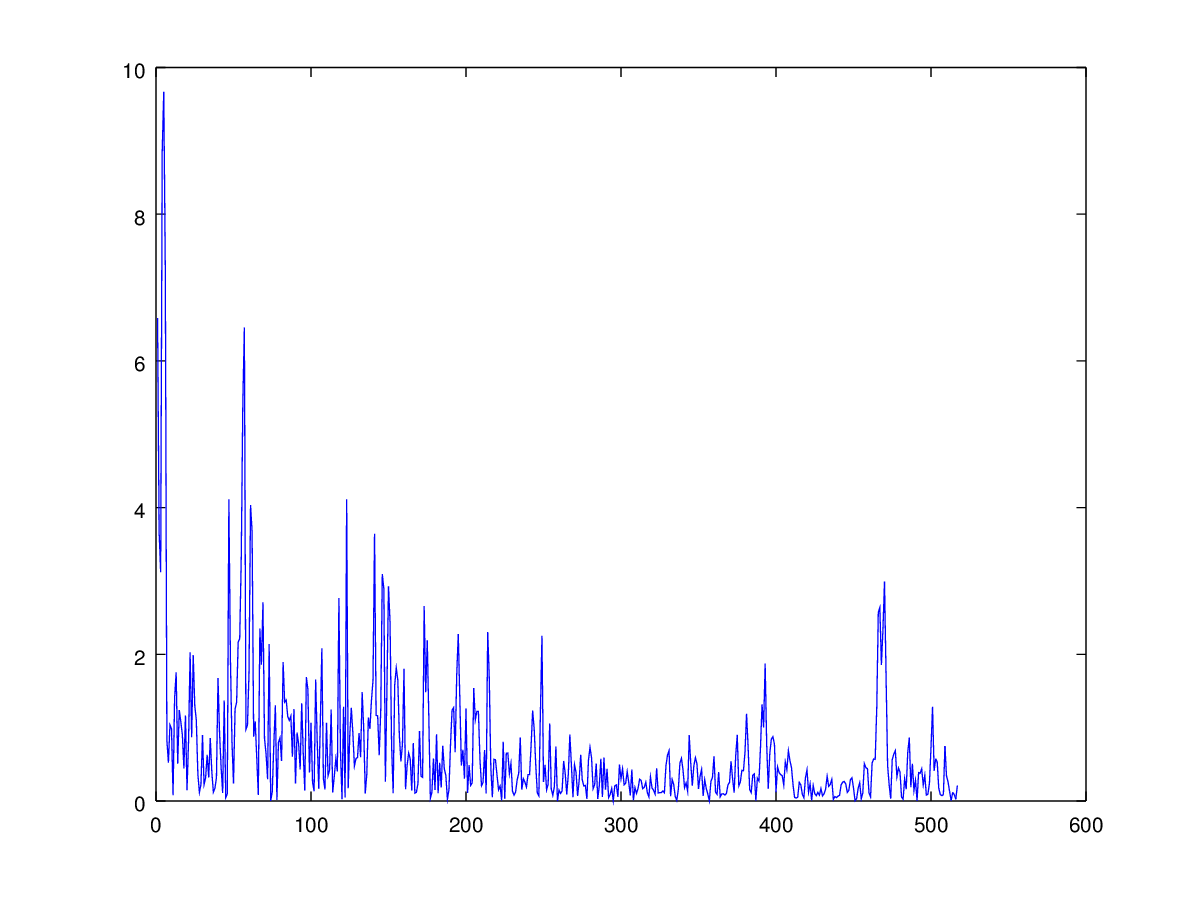}}
	\quad
	\subfigure[Inverse of benchmark portfolio RMSE\label{subfig: Benchmark portfolio RMSE}]
	{\includegraphics[scale=0.31]{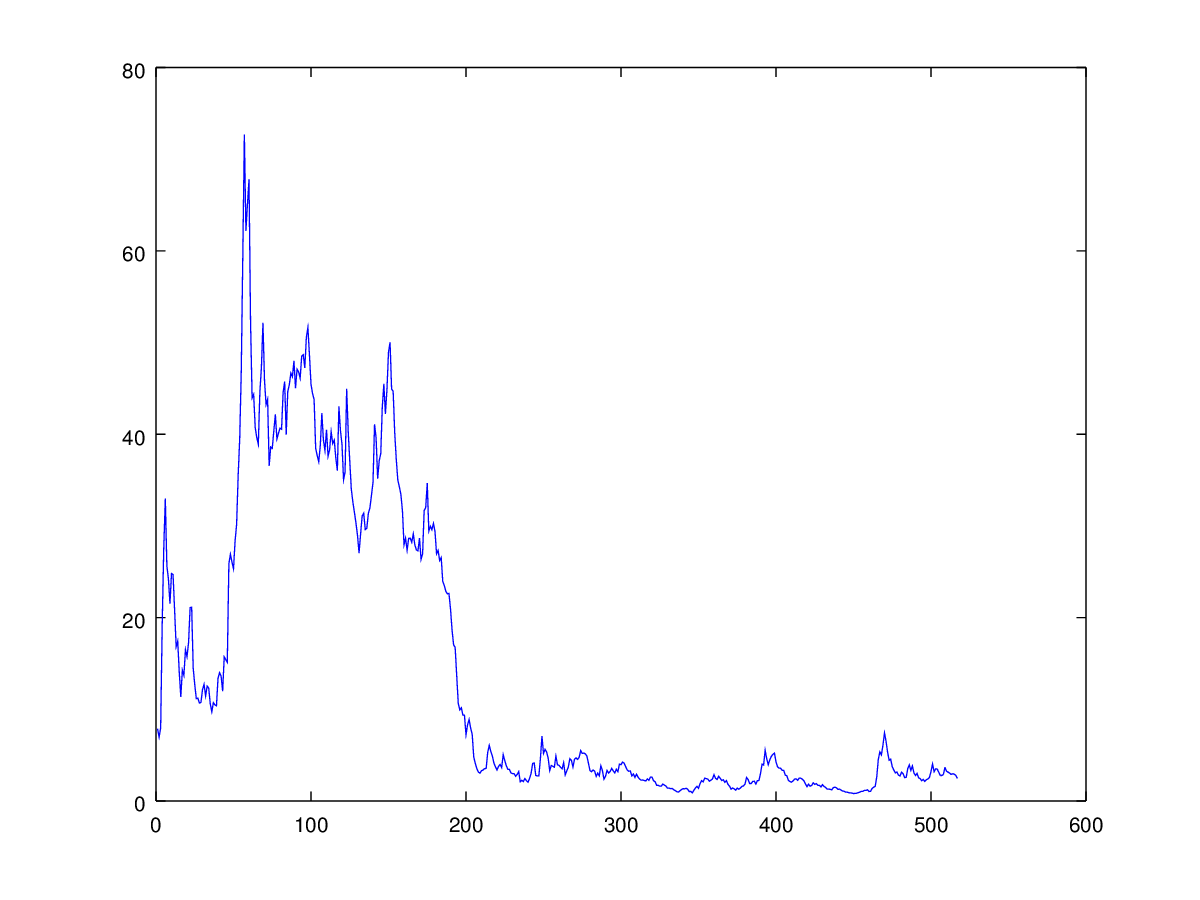}}
	\subfigure[State variable component $X$\label{subfig: State variable component X}]
	{\includegraphics[scale=0.31]{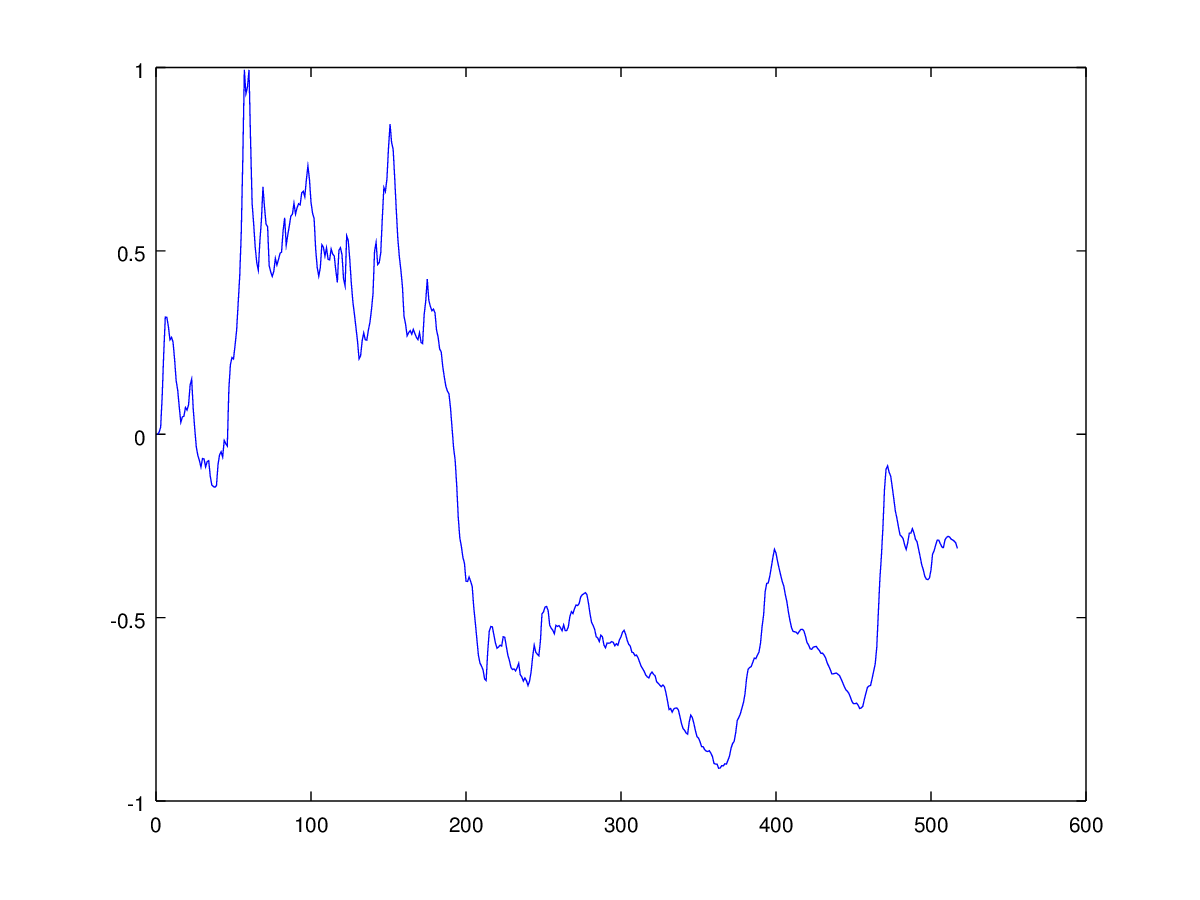}}
	\quad
	\subfigure[Short rate $r$\label{subfig: Short rate}]
	{\includegraphics[scale=0.31]{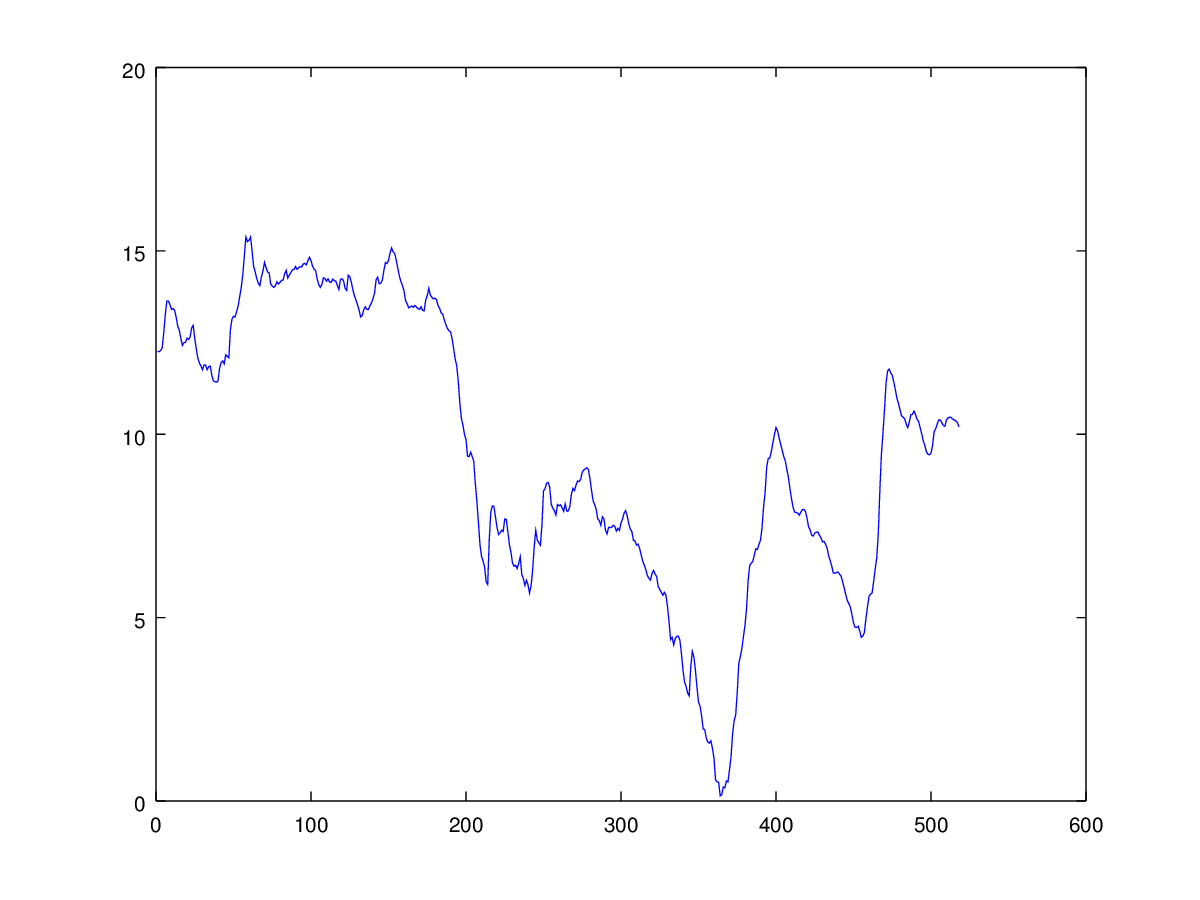}}
	\caption{Benchmark portfolio data and fit.}\label{fig: Benchmark portfolio data and fit}
\end{figure}

Figure \ref{fig: Benchmark portfolio data and fit} \subref{subfig: Benchmark portfolio data and fit} plots the observed inverse of benchmark portfolio data and the fit produced by Kalman filter. Figure \ref{fig: Benchmark portfolio data and fit} \subref{subfig: Benchmark portfolio difference} and \subref{subfig: Benchmark portfolio RMSE}, with unit in basis point, shows respectively the pricing error generated by Kalman filter and the root mean square pricing error (RMSE) computed over 100 Monte Carlo replications. Figure \ref{fig: Benchmark portfolio data and fit} \subref{subfig: State variable component X} displays time series of estimated state variable component $X$ underlying the benchmark portfolio dynamics and taking value in the compact interval $[-1, 1]$. Figure \ref{fig: Benchmark portfolio data and fit} \subref{subfig: Short rate} displays time series of estimated short rate $r$ adjusted by the level parameter $\alpha$ and thus taking only positive values. The simulated samples consist of 517 monthly time series observations from January 1970 to January 2013.\\
We note that one dimensional component $X$, with a mean RMSE equal to 15.24 bps, is already good enough to explain the inverse of benchmark portfolio dynamics structure and to produce a reasonable fit to the observed data. In particular the fit behaves better in the tail, which is a desirable situation as we are fitting the inverse of LLMA world index value.

\begin{figure}[!htbp]
	\centering
	\subfigure[Longevity index data and fit\label{subfig: Longevity index data and fit}]
	{\includegraphics[scale=0.31]{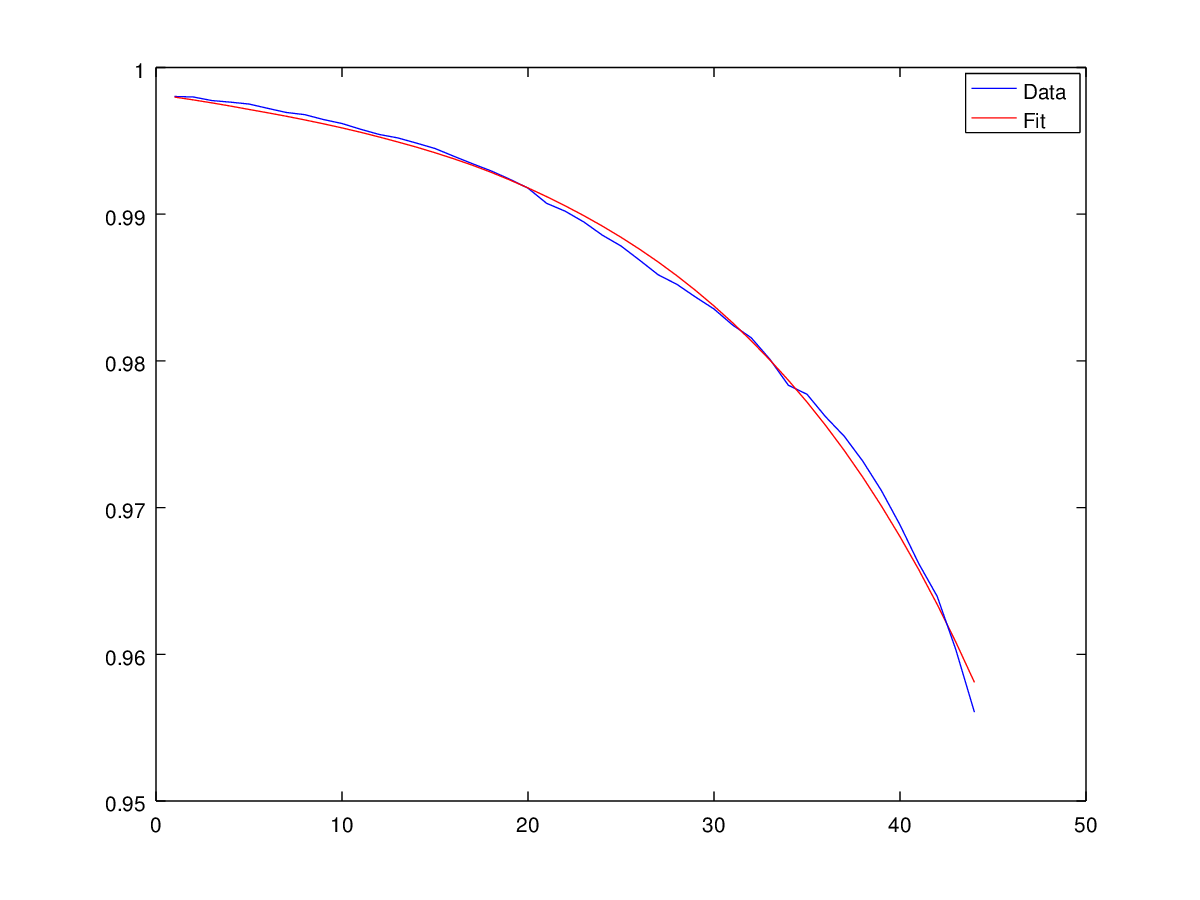}}
	\\ \ \\
	\subfigure[Pricing error generated by Kalman filter\label{subfig: Longevity index difference}]
	{\includegraphics[scale=0.31]{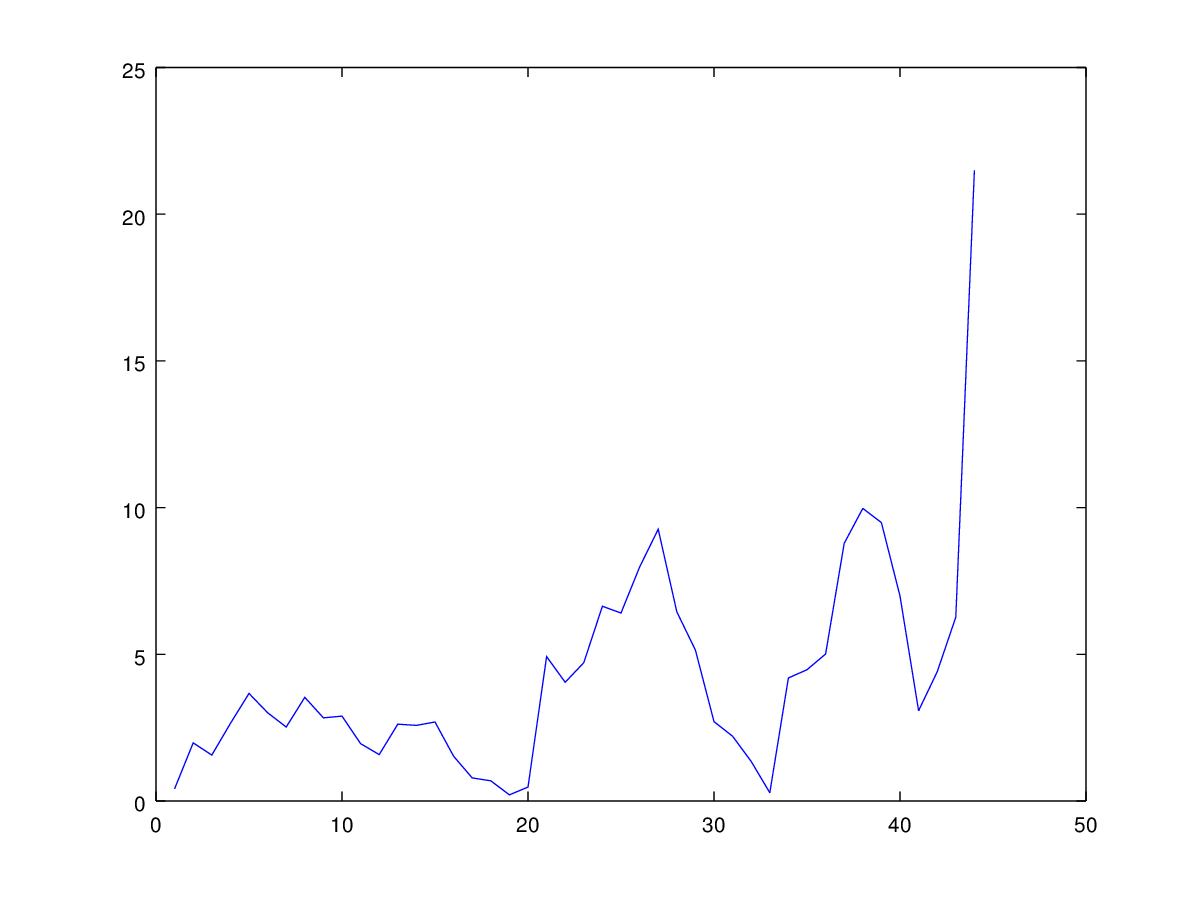}}
	\quad
	\subfigure[Longevity index RMSE\label{subfig: Longevity index RMSE}]
	{\includegraphics[scale=0.31]{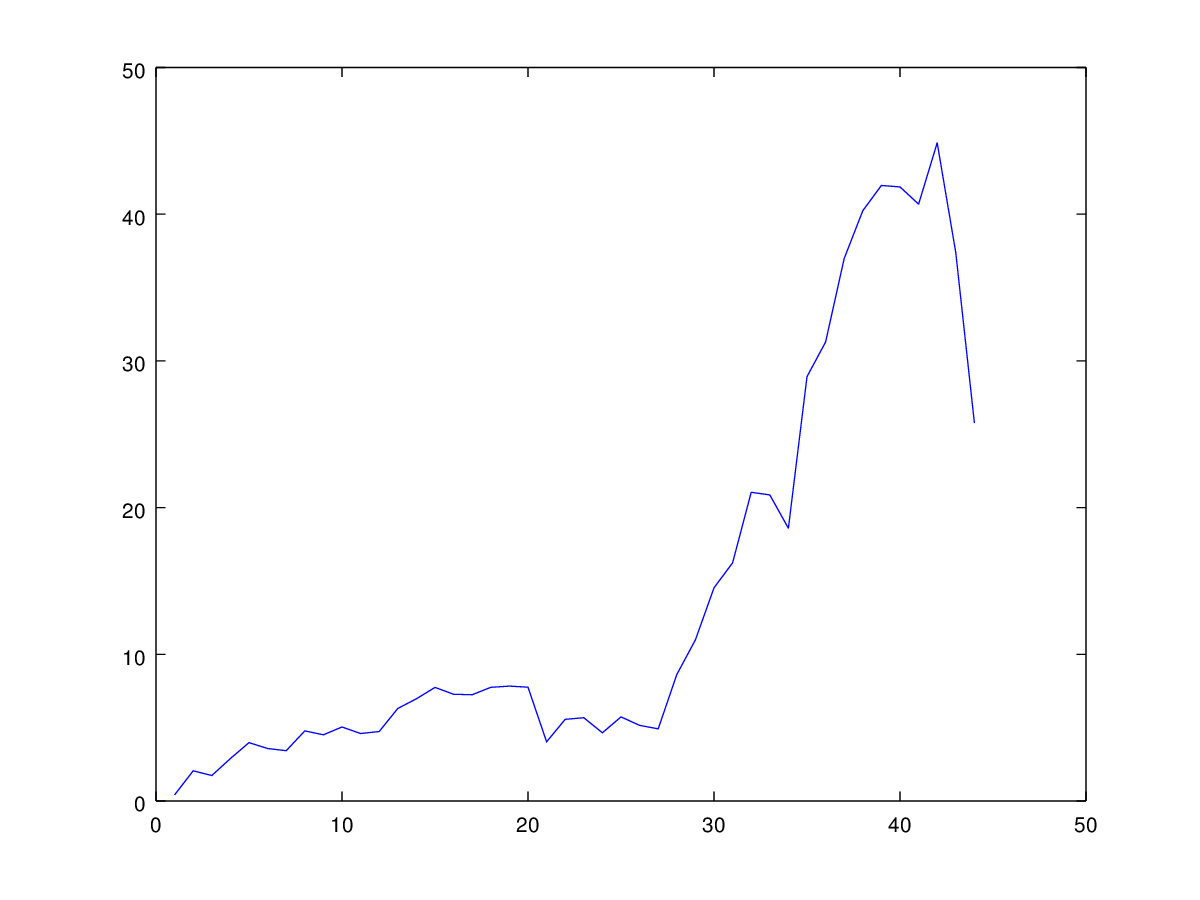}}
	\subfigure[State variable component $Y$\label{subfig: State variable component Y}]
	{\includegraphics[scale=0.31]{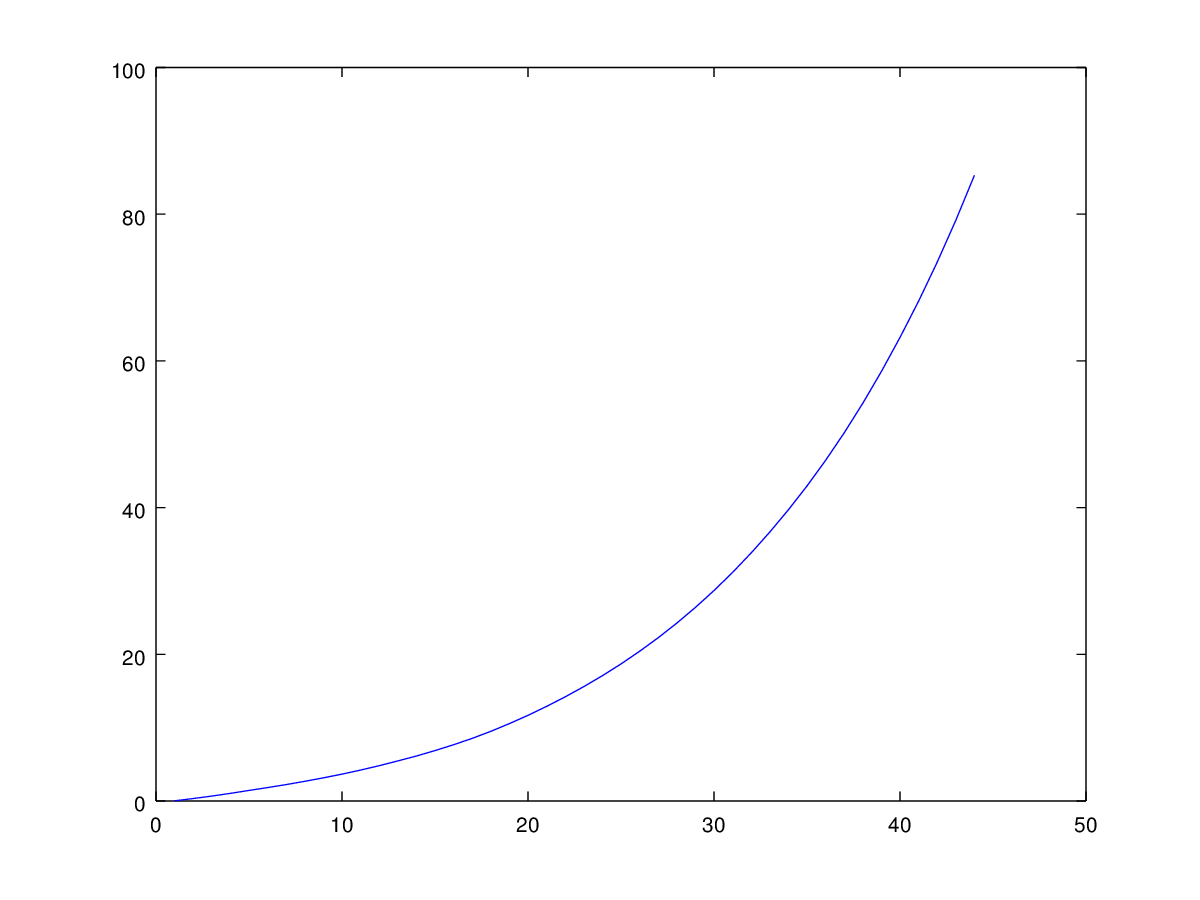}}
	\quad
	\subfigure[Mortality intensity $\mu$\label{subfig: Mortality intensity}]
	{\includegraphics[scale=0.31]{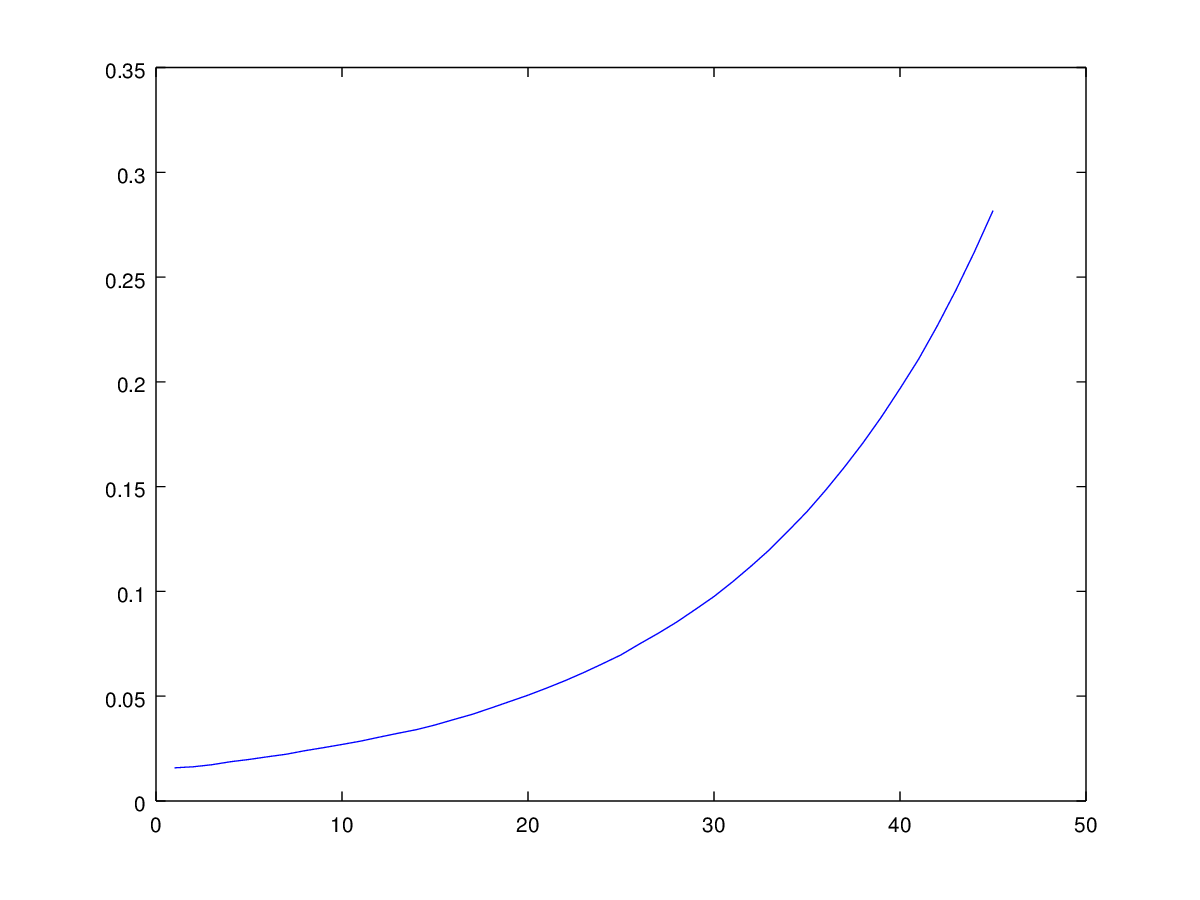}}
	\caption{Longevity index data and fit.}\label{fig: Longevity index data and fit}
\end{figure}

Similarly Figure \ref{fig: Longevity index data and fit} \subref{subfig: Longevity index data and fit} plots the observed longevity index data and the fit produced by (\ref{eq: approx Y}) with estimated parameter sets. Figure \ref{fig: Longevity index data and fit}\subref{subfig: Longevity index difference} and \subref{subfig: Longevity index RMSE}, with unit in basis point, show respectively the pricing error associated to (\ref{eq: approx Y}) and the root mean square pricing error (RMSE) computed over 100 Monte Carlo replications. Figure \ref{fig: Longevity index data and fit}\subref{subfig: State variable component Y} displays time series of estimated state variable component $Y$. Figure \ref{fig: Longevity index data and fit}\subref{subfig: Mortality intensity} displays time series of estimated mortality intensity $\mu$.
The simulated samples consist of 44 annual time series observations from January 1970 to January 2013.\\
The absence of the diffusion term in the $Y$ dynamics produces rather smooth paths for $Y$ and consequently for the longevity index fit. This is reasonable as longevity index data has very light oscillations along the trend, with a mean RMSE value of 15.39 bps. However, the poor data set of the longevity index, with always less than 50 annual observations for one age cohort, and the long time frame between two consecutive data can be a drawback for calibration.

\section{Conclusion}

We consider a polynomial diffusion state variable on compact state space to model the hybrid financial-insurance market. We include a dependence structure between OIS short rate and mortality intensity and are able to guarantee their positivity when needed. The model is parsimonious and analytically tractable. By following the benchmark approach we compute optimal premium as well as risk-minimizing strategy for life insurance liabilities. We complete our theoretical analysis by a calibration of a particular model specification to real data. As showed in Section \ref{sec: simulation}, we obtain a good fitting to the MSCI and LLMA index also in the parsimonious case of a two-dimensional state variable under linear polynomial specification. 

\appendix
\section{Benchmarked risk-minimization for payment streams}\label{app: benchmarked risk-min payment process}

We give here the general definitions and results of the benchmarked risk-minimizing method for payment streams by following mainly \cite{Bia-Cre} and \cite{Bar} for the presentation.  

We introduce first the following Hilbert spaces
\[
	M_0^2(\P, \Gb) := \left. \{ M:= (M_t)_{t \in [0,T]} \ \Gb\text{-martingale } \right| \ M_0 = 0, \ \sup_{t \in [0,T]} \esp{M_t^2} < \infty \}, 
\]
\[
	L^2(\tilde S, \P, \Gb) = \left. \bigg\{  \ \delta \ \ \R^l\text{-valued } \Gb\text{-predictable processes } \right| \ \mathbb{E}\left[ \int^T_0 \delta^\top_u \ud[\tilde S]_u \delta_u \right]  < \infty \bigg\}\footnote{$[\tilde S] = ([\tilde S^i,\tilde S^j])_{i,j=1,...,l}$ denotes the quadratic variation matrix process of the benchmarked primary security vector $\tilde S$.},
\]
\[
	\mathcal{I}(\tilde S, \P, \Gb) =\left \{ \left. \int^T_0 \delta_u^\top \ud \tilde S_u \right| \delta \in L^2(\tilde S, \P, \Gb) \right \},
\]
with norms given respectively by
\[
	\norm{ M }_{M_0^2(\P, \Gb)} := \sup_{t \in [0,T]} \esp{M_t^2}^{\frac{1}{2}} = \esp{M_T^2}^{\frac{1}{2}}= \sup_{t \in [0,T]} \esp{[M]_t}^{\frac{1}{2}} = \esp{[M]_T}^{\frac{1}{2}},
\]
\[
	\norm{\delta}_{L^2(\tilde S, \P, \Gb)} := \left( \mathbb{E}\left[ \int^T_0 \delta^\top_u \ud[\tilde S]_u \delta_u \right] \right)^{\frac{1}{2}},
\]
\[
	\norm{\int^T_0 \delta_u^\top \ud \tilde S_u}_{\mathcal{I}(\tilde S, \P, \Gb)} := \esp{\left(\int^T_0 \delta_u^\top \ud \tilde S_u\right)^2}^{\frac{1}{2}},
\]
for $M \in M_0^2(\P, \Gb)$ and $\delta \in L^2(\tilde S, \P, \Gb)$. Lemma 3.4 of \cite{Bar} (or Lemma 2.1 of \cite{Sch}) shows that $\mathcal{I}(\tilde S, \P, \Gb)$ is a stable subspace of $M_0^2(\P, \Gb)$. In particular, for every $\delta \in L^2(\tilde S, \P, \Gb)$ we have
\[
	\norm{\delta}_{L^2(\tilde S, \P, \Gb)} = \norm{ \int^T_0 \delta_u^\top \ud \tilde S_u }_{\mathcal{I}(\tilde S, \P, \Gb)} = \norm{\int^T_0 \delta_u^\top \ud \tilde S_u} _{M_0^2(\P, \Gb)}.
\]

\begin{defn}
	An $L^2$-\emph{admissible strategy} is a process $\delta := (\delta_t)_{t \in [0,T]}$ such that
	\begin{itemize}
		\item[(1)] $\delta \in L^2(\tilde S, \P, \Gb)$,
		\item[(2)] the associated benchmarked value process $\tilde S^\delta$ with
		\[
			\tilde S^\delta_{t-} := \delta_t^\top \tilde S_t, \ \ \ \ t \in [0,T]
		\] 
		is in $M_0^2(\P, \Gb)$.
	\end{itemize}
\end{defn}

Now we fix a process $A$ as defined in (\ref{eq: A process definition}) which models the benchmarked cumulative payments towards a policyholder.

\begin{defn}
	The \emph{benchmarked cumulative cost process} $C^\delta:=(C_t^\delta)_{t \in [0,T]}$ of a $L^2$-admissible strategy $\delta$ associated to $A$ is defined by
	\[
		C_t^\delta = \tilde S^\delta_t - \int_0^t \delta_u \ud \tilde S_u + A_t, \ \ \ \ t \in [0,T].
	\]
\end{defn}

\begin{defn}
	The \emph{risk process} $R^\delta:=(R^\delta_t)_{t \in [0,T]}$ of an $L^2$-admissible strategy $\delta$ is defined by
	\[
		R_t^{\delta} = \condespg{(C_T^{\delta} - C_t^{\delta})^2},\ \ \ \ t \in [0,T].
	\]
\end{defn}

\begin{defn}
	An $L^2$-admissible strategy $\bar \delta$ is called \emph{benchmarked risk-minimizing} for $A$ if 
	\begin{itemize}
		\item[(1)] $\tilde S_T^{\bar \delta} = 0$ $\P$-a.s.,
		\item[(2)] $R_t^{\bar \delta} \leqslant R_t^{\delta}$ $\P$-a.s. for every $t \in [0,T]$ and for any $L^2$-admissible strategy $\delta$ such that $\tilde S_T^{\bar \delta} = \tilde S_T^{\delta}$ $\P$-a.s., $\bar \delta_u = \delta_u$ $\P$-a.s. for all $u \leqslant t$.
	\end{itemize}
\end{defn}

\begin{lemma}\label{lemma: benchmarked cost martingale}
	The benchmarked cumulative cost process of any benchmarked risk-minimizing strategy is a $(\P, \Gb)$-martingale.
\end{lemma}

\begin{proof}
	This lemma is a combination of Lemma 3.5 of \cite{Bia-Cre} and Lemma A.4 of \cite{Mol-i}.
\end{proof}

\begin{lemma}\label{lemma: benchmarked risk-minimizing value process}
	The benchmarked value process associated to a benchmarked risk-minimizing strategy $\bar \delta$ for $A$ is given by
	\[
		\tilde S_t^{\bar \delta} = \condespg{A_T - A_t}, \ \ \ \ t \in [0,T].
	\]
\end{lemma}

\begin{proof}
	The proof is a straightforward consequence of Lemma \ref{lemma: benchmarked cost martingale} and Lemma 3.12 of \cite{Mol-i}.
\end{proof}

The following theorem gives the solution of the benchmarked risk-minimizing problem.
\begin{theorem}
	Let 
	\begin{equation}\label{eq: GKW decomposition}
		A_T = \mathbb{E}[A_T] + \int_0^T \left( \delta^A_u \right)^\top \ud \tilde S_u + L^A_T, \ \ \ \ \P-\text{a.s.},
	\end{equation}
	be the Galtchouk-Kunita-Watanabe decomposition\footnote{See \cite{Ans-Str} for an overview of Galtchouk-Kunita-Watanabe decomposition.} of $A_T$,
	where $\int_0^T \delta^A_u \ud \tilde S_u$ is the projection of $\left(A_T - \mathbb{E}[A_T]\right)$ on the space $\mathcal{I}(\tilde S, \P, \Gb)$ with $\delta^A \in L^2(\tilde S, \P, \Gb)$ and $L^A \in M_0^2(\P, \Gb)$ is $\P$-strongly orthogonal to $\mathcal{I}(\tilde S, \P, \Gb)$. 
	The unique benchmarked risk-minimizing strategy $\bar \delta$ for $A$ is given by $\bar \delta = \delta^A$. The associated benchmarked cumulative cost process is 
	\[
		C_t^{\bar \delta} = \esp{A_T} + L^A_t = C_0^{\bar \delta} + L^A_t, \ \ \ \ t \in [0,T],
	\]
	and the benchmarked value process is
	\[
		\tilde S_t^{\bar \delta} = \condespg{A_T - A_t}, \ \ \ \ t \in [0,T].
	\]
\end{theorem}

\begin{proof}
	The theorem follows from Lemma \ref{lemma: benchmarked cost martingale} and Theorem 2.1 of \cite{Mol-i}.
\end{proof}

As the classic risk-minimizing strategy, the crucial point of the solution of the benchmarked risk-minimizing strategy is finding the Galtchouk-Kunita-Watanabe decomposition (\ref{eq: GKW decomposition}). 

We stress that, the orthogonal projection provided by the decomposition (\ref{eq: GKW decomposition}) shows that every benchmarked cumulative payment $A_T$ has a perfectly hedgeable part  $\int_0^{T} \left( \delta^A_u \right)^\top \ud \tilde S_u$ and a totally unhedgeable part $\left( \esp{A_T} + L^A_T \right)$ covered by the benchmarked cumulative cost process $C$. Furthermore, according to Lemma \ref{lemma: benchmarked risk-minimizing value process}, the benchmarked value process associated to the unique benchmarked risk-minimizing strategy $\bar \delta$ for a given benchmarked cumulative payment process $A$ coincides with the benchmarked value of the real world pricing formula given in (\ref{eq: pricing formula dividend}), i.e.
\[
	\tilde S_t^{\bar \delta} = \frac{V_t}{S^*_t}, \ \ \ \ t \in [0, T].
\]
The benchmarked hedging problem and its relation with the real world pricing formula have already been studied in \cite{Bia} and \cite{Bia-Cre} in the case of a $T$-claims $\bar D$, i.e. when the dividend process $D$ is given by
\[
	D_t = \Ind{t = T} \bar D, \ \ \ \ t \in [0,T],
\]
with $\bar D$ a square integrable $\G$-measurable random variable. The real world pricing formula is reduced to 
\[
	V_t = S_t^* \condespg{\frac{\bar D}{S^*_T}}, \ \ \ \  t \in [0,T[,
\]
which is the original definition of fair price given in e.g. \cite{Pla3} for a $T$-claim $\bar D$. In this case, if the $T$-claim admits a self-financing strategy, then thanks to the supermartingale property of the benchmark portfolio in Definition (\ref{def: benchmark}), $V$ corresponds to the least expensive self-financing portfolio which replicates $\bar D$.

\bibliographystyle{plain}
\bibliography{bibliog}

\end{document}